\documentclass[11pt, usenatbib, usegraphicx]{article}
\usepackage[margin=1in]{geometry}
\usepackage{amsmath,amsthm,amssymb,enumerate,graphicx,subcaption}
\usepackage{bbm}
\usepackage{setspace}
\usepackage{thmtools}
\usepackage{thm-restate}
\usepackage{hhline}
\usepackage[ruled,vlined]{algorithm2e}
\usepackage{enumerate}
\usepackage[mathscr]{euscript}
\usepackage[shortlabels]{enumitem}
\usepackage{prodint}
\usepackage[hidelinks]{hyperref}
\usepackage{cleveref,color}
\usepackage{nameref}
\usepackage{changepage}
\usepackage{xcolor}
\usepackage{caption}
\usepackage{subcaption}
\usepackage{siunitx}
\usepackage{booktabs}
\usepackage{multicol}
\usepackage{threeparttable}
\usepackage{natbib}

\declaretheorem[name=Theorem]{thm}
\declaretheorem[name=Proposition]{prop}
\declaretheorem[name=Lemma]{lemma}

\makeatletter
\def\namedlabel#1#2{\begingroup
    #2%
    \def\@currentlabel{#2}%
    \phantomsection\label{#1}\endgroup
}
\makeatother

\DeclareMathOperator*{\inprob}{\stackrel{P}{\longrightarrow}}
\newcommand{\bounded}{O_{\mathrm{P}}}
\newcommand{\fasterthan}{o_{P}}
\newcommand{\boundeddet}{O}

\def\ci{\perp\!\!\!\perp}

\renewcommand{\d}[1]{\mathbb{#1}}

\newcommand{\n}[1]{\mathrm{#1}}

\title{Nonparametric Sensitivity Analysis for Unobserved Confounding with Survival Outcomes}

\author{
\begin{tabular}{@{}p{0.44\textwidth}@{\hspace{-0.1cm}}p{0.55\textwidth}@{}}
\centering
Rui Hu \\ 
School of Artificial Intelligence \\ 
Shenzhen Technology University \\ 
hurui@sztu.edu.cn
&
\centering
Ted Westling \\ 
Department of Mathematics and Statistics \\ 
University of Massachusetts Amherst \\ 
twestling@umass.edu
\end{tabular}
}

\begin{document}

\maketitle

\begin{abstract}
In observational studies, the observed association between an exposure and outcome of interest may be distorted by unobserved confounding. Causal sensitivity analysis can be used to assess the robustness of observed associations to potential unobserved confounding. For time-to-event outcomes, existing sensitivity analysis methods rely on parametric assumptions on the structure of the unobserved confounders and Cox proportional hazards models for the outcome regression. If these assumptions fail to hold, it is unclear whether the conclusions of the sensitivity analysis remain valid. Additionally, causal interpretation of the hazard ratio is challenging. To address these limitations, in this paper we develop a nonparametric sensitivity analysis framework for time-to-event data. Specifically, we derive nonparametric bounds for the difference between the observed and counterfactual survival curves and propose estimators and inference for these bounds using semiparametric efficiency theory. We also provide nonparametric bounds and inference for the difference between the observed and counterfactual restricted mean survival times. We demonstrate the performance of our proposed methods using numerical studies and an analysis of the causal effect of elective neck dissection on mortality in patients with high-grade parotid carcinoma.
\end{abstract}

\doublespacing

\newpage

\section{Introduction} \label{sec:introduction}
A common goal of cohort studies is to assess the causal effect of a baseline treatment or exposure on a time-to-event outcome using observational data. Time-to-event regression models, such as Cox proportional hazards regression \citep{cox1972regression}, are commonly used to investigate the covariate-adjusted association between an exposure and a time-to-event outcome. However, such models often summarize the association as a hazard ratio, and hazard ratios are difficult to interpret causally due to potential imbalance between treated and control units who are still at risk at time $t > 0$, even when the treated and control groups are balanced at time $t = 0$ \citep{hernan2010hazards,martinussen2020subtleties}.

Alternatively, the average causal effect (ACE) of a binary baseline treatment on a time-to-event outcome can be measured in terms of the difference between the marginal counterfactual survival curves had the population been assigned to receive the treatment or not \citep{maldonado2002estimating,hernan2004definition,hernan2010hazards,martinez2021summarizing}. If recorded baseline covariates contain all confounders of the exposure-outcome relationship, known as the \textit{no unobserved confounding} assumption, the ACE is identified with the difference in covariate-adjusted survival curves \citep{beran1981nonparametric,robins1986new,gill2001causal}. Several estimators of the covariate-adjusted survival curves have been proposed \citep{dabrowska1989uniform,hubbard2000nonparametric,zeng2004estimating,zhang2012contrasting,bai2013doubly,westling2023inference}. However, the no unobserved confounding assumption is generally untestable and does not necessarily hold in practice. When violated, the difference in covariate-adjusted survival curves  may not reflect the true causal effect.

In causal inference, \emph{sensitivity analysis} is used to assess how sensitive observed effects are to violations of the no unobserved confounding assumption. In general, a sensitivity analysis can be conducted in four steps:  1) identify bounds on the ACE as a function of the observed data distribution and \emph{sensitivity parameters} that quantify the amount of unobserved confounding; 2) obtain inference for the bounds using the observed data under given values of the sensitivity parameters; 3) infer the minimum magnitude of the sensitivity parameters that would change the scientific conclusions; and 4) assess the plausibility of this minimal magnitude in the scientific context \citep{cinelli2020making,mcclean2024calibrated}.

There is an extensive history of research on sensitivity analysis for binary and continuous outcomes---see, e.g., \cite{cornfield1959smoking,rosenbaum1983assessing,imbens2003sensitivity,liu2013introduction,richardson2015nonparametric,carnegie2016assessing,zhang2022semi,nabi2024semiparametric} and references therein. Fewer methods for sensitivity analysis accommodate right-censored time-to-event outcomes. \cite{huang2020sensitivity} extended the model-based sensitivity analysis framework \citep{rosenbaum1983assessing,imbens2003sensitivity} to survival outcomes under parametric assumptions on the data generating process. Their approach requires the unobserved confounder to be binary and independent of other observed covariates, and assumes Cox and probit models for the outcome regression and propensity score, respectively. \cite{ding2016sensitivity} derived an approximate bound on the hazard ratio in terms of the maximal relative risk of treatment on the unobserved confounders and the maximal hazard ratio of unobserved confounders on the outcome without invoking any assumption on the structure of unobserved confounding. However, the bound only applies to rare time-to-event outcomes, and inferential results for the bound have not been discussed. \cite{lu2023flexible} proposed a flexible framework using outcome regression ratios as sensitivity parameters, and suggested that their framework can be extended to survival outcomes. However, this extension and associated inferential results were not discussed in detail. 

In this article, we develop a nonparametric sensitivity analysis framework for unobserved confounding with right-censored time-to-event outcomes.  Our causal parameter of interest is the difference between marginal counterfactual survival curves. Our sensitivity parameters are the proportion of residual variance of functions of the outcome regression and treatment propensity explained by the unobserved confounders. For fixed values of sensitivity parameters and time, we propose a cross-fitted estimator of the effect bounds that permits the use of data-adaptive nuisance estimators, and we use our estimator to conduct pointwise and uniform sensitivity analysis. We also propose inference for bounds on the difference in counterfactual restricted mean survival times. Finally, we extend robustness values and benchmarking to the time-to-event setting for assessing the results of a sensitivity analysis.

Our approach has the following merits: (1) our target parameter is interpretable as a causal effect without strong parametric or semiparametric assumptions; (2) our framework does not require parametric or semiparametric assumptions on the structure of unobserved confounders, the outcome regression function, or the treatment assignment mechanism; (3) our estimators can converge at the parametric rate and our inference procedures can be asymptotically valid even when flexible nonparametric methods are used to estimate the survival outcome regression and treatment assignment functions, provided these estimators achieve sufficient rates of convergence; and (4) our sensitivity parameters are interpretable without parametric assumptions, and their plausibility can be assessed by comparing with observed confounding. To the best of our knowledge, no existing method for time-to-event data meets all these criteria. Therefore, our work fills an important gap by providing flexible tools for causal sensitivity analysis  in observational studies with right-censored time-to-event outcomes.

The remainder of the paper is organized as follows. In Section~\ref{sec:setup}, we introduce the data structure, causal effect of interest and its identification, bounds under unobserved confounding, and interpretation of the sensitivity parameters. In Section~\ref{sec:estinf}, we propose estimators and inference for the bounds. In Section~\ref{sec:senspar}, we discuss robustness values and benchmarking. In Section~\ref{sec:simulation}, we assess the behavior of our methods using numerical studies, and in Section~\ref{sec:data}, we use our methods to assess the robustness of the effect of elective neck dissection on mortality in patients with high-grade parotid carcinoma. Proofs of all theorems are provided in Supplementary Material.

\section{Sensitivity analysis framework} \label{sec:setup}
\subsection{Notation and statistical setting}
We let $A \in \{0,1\}$ be a binary treatment and $W \in \mathbb{R}^p$ be a vector of observed pre-treatment covariates. Both $A$ and $W$ are assumed to be recorded prior to time $t=0$. We adopt the potential outcomes framework \citep{neyman1923applications,rubin1974estimating} to define $T(a) \in (0, \infty]$ and $C(a) \in [0, \infty]$ as potential event and  right-censoring times, respectively, that would have been observed under treatment assignment $A=a$. We define $O_c:=(W, A, T(0), T(1), C(0), C(1))$ as the full (unobservable) causal data unit and $P_c$ as its distribution. If the treatment is unique and each unit's treatment is independent of all other units' outcomes, called the stable unit treatment value assumption \citep{rubin1980randomization}, we can denote the factual event and censoring times by $T:=T(A)$ and $C:=C(A)$ respectively. In the time-to-event setting,  the event time of interest $T$ is often not fully observed, but is instead right censored by $C$. We then observe the right-censored time $Y:=\min(T,C)$ and the censoring indicator $\Delta:=I(T\leq C)$ for each unit. We assume we observe $n$ independent and identically distributed observations $(O_1, \dotsc ,O_n)$ of the observed data unit $O:=(Y, \Delta, A, W)$ drawn from an unknown distribution $P$. We denote parameters that depend on the distribution $P_c$ of the causal data with a subscript $c$ and parameters that only depend on the distribution $P$ of the observed data with a subscript $P$.

\subsection{Causal effect of interest and its identification}
Our causal effect of interest is the counterfactual survival difference at time $t \in (0,\infty)$:
\begin{equation}\label{eq:par1}
    \theta_c(t) := P_c(T(1)>t) - P_c(T(0)>t)\,.
\end{equation}
We are interested in the scalar causal effect $\theta_c(t) \in [-1,1]$ for fixed $t$ and the curve of causal effects $\{t \mapsto \theta_c(t) : t \in (0, \zeta]\}$ for some fixed $\zeta < \infty$. We introduce the following causal identification conditions.
\begin{description}[style=multiline,leftmargin=1.4cm, labelindent=.3cm]
    \item[\namedlabel{itm:idf_1}{(A1)}] $T(a) I(T(a) \leq t) \ci A \mid W$ (no unobserved confounding between treatment and event);
    \item[\namedlabel{itm:idf_2}{(A2)}] $C(a) I(C(a) \leq t) \ci A \mid W$ (no unobserved confounding between treatment and censoring);
    \item [\namedlabel{itm:idf_3}{(A3)}] $T(a) I(T(a) \leq t) \ci C(a) I(C(a) \leq t) \mid A=a, W$ (conditionally independent censoring);
    \item [\namedlabel{itm:idf_4}{(A4)}] $P(P(A=a \mid W)>0) = 1$  (treatment assignment positivity); and
    \item [\namedlabel{itm:idf_5}{(A5)}] $P(P(C \geq t \mid A = a, W)>0) = 1$ (censoring probability positivity).
\end{description} 
If \ref{itm:idf_1}--\ref{itm:idf_5} hold for each $a \in \{0,1\}$, then the causal parameter $\theta^c(t)$ is identified via the \emph{backdoor formula}:
\begin{equation}\label{eq:par2}
\theta_c(t) = \theta_P(t) := E[S_P(t\mid A=1, W)] - E[S_P(t\mid A=0, W)]\,,
\end{equation}
where $S_P(t\mid A, W):=P(T>t \mid A, W)$ is identified using the product integral \citep{beran1981nonparametric,robins1986new,dabrowska1989uniform,gill2001causal,westling2023inference}. We refer to $\theta_P(t)$ as the \emph{observed-data effect} because it is defined in terms of the distribution $P$ of the observed data, in contrast to the causal effect $\theta_c(t)$, which is defined in terms of the distribution of the full data.

Assumption~\ref{itm:idf_1} requires that there is no unobserved confounding between treatment and outcome. This is often violated in observational studies because when the treatment assignment mechanism is unknown, it is typically impossible to guarantee that all confounders have been measured. As a result, in observational studies the identification result~\eqref{eq:par2} may not hold, and the observed-data effect $\theta_P(t)$ is not necessarily equal to the causal effect $\theta_c(t)$. In this article, we develop sensitivity analysis methods for assessing the potential impact of unobserved confounding on the difference between $\theta_P(t)$ and $\theta_c(t)$. We note that violation of the other identification assumptions \ref{itm:idf_2}--\ref{itm:idf_5} can also invalidate~\eqref{eq:par2}. Here, we focus on unobserved confounding between the treatment and outcome.

\subsection{Bounds under unobserved confounding}
In this section, we provide bounds on the causal effect $\theta_c(t)$ in the presence of unobserved confounding. We assume there exists an unobserved $U \in \mathbb{R}^{q}$ such that assumptions \ref{itm:idf_1}--\ref{itm:idf_5} are satisfied with $(W,U)$ in place of $W$. We then have 
\begin{equation}\label{eq:par3}
    \theta_c(t) = E[S_c(t\mid A=1, W, U)] - E[S_c(t\mid A=0, W, U)],
\end{equation}
where $S_c(t \mid a, w,u) := P_c(T > t \mid A = a, W = w, U = u)$. We define $g_{c,t}(a,w,u) := S_c(t\mid a, w, u)$ and $\alpha_c(a,w,u) := \frac{a}{\pi_c(w,u)} - \frac{1-a}{1-\pi_c(w,u)}$ where $\pi_c(w,u) := P_c(A=1 \mid W=w, U=u)$. In addition, we define $g_{P,t}(a,w) := S_P(t\mid a, w)$ and $\alpha_P(a,w) := \frac{a}{\pi_P(w)} - \frac{1-a}{1-\pi_P(w)}$, where $\pi_P(w) := P(A=1 \mid W=w)$. For notational simplicity, we set $g_{c,t} := g_{c,t}(A,W,U)$, $g_{P,t}:=g_{P,t}(A,W)$, $\alpha_c := \alpha_c(A,W,U)$, and $\alpha_P := \alpha_P(A,W)$.
We then have the following decomposition of $[\theta_P(t) - \theta_c(t)]^2$.
\begin{prop}\label{prop:par_id}
It holds that
\begin{align*}
\left[\theta_P(t) - \theta_c(t)\right]^2 =  \psi_P(t)\tau_P \left[ \rho_c(t)\right]^2 s_{c,T}(t) \frac{s_{c,A}}{1-s_{c,A}}\,,
\end{align*}
where $\psi_P(t) := E\left[ g_{P,t} \left\{1-g_{P,t}\right\}\right]$, $\tau_P := E \left[\alpha_P^2\right] = E\left[1/ \left\{\pi_P[1-\pi_P]\right\}\right]$\,, 
\begin{align*}
    \rho_c(t) &:= \n{Cor}\left(g_{c,t}-g_{P,t}, \: \alpha_c-\alpha_P\right)\,, \\
    s_{c,T}(t) &:= \frac{\n{Var}(g_{c,t}) - \n{Var}(g_{P,t})}{\n{Var}(I(T>t)) -  \n{Var}(g_{P,t})}\,, \text{ and} \\
    s_{c,A} &:= 1 - \frac{E \left[\alpha_P^2\right]}{E \left[\alpha_c^2\right]} 
\end{align*}
and where we interpret $0/0$ as 0 if $\psi_P(t) = 0$ or $\tau_P = 0$. 
Therefore,
\begin{align*}
\theta_{P,l}\left(t, s_{c,T}(t)s_{c,A} / [1 - s_{c,A}]\right) &\leq \theta_c(t) \leq \theta_{P,u}\left(t, s_{c,T}(t)s_{c,A} / [1 - s_{c,A}]\right) \quad \text{ for} \\
    \theta_{P,l}(t, v) &:= \theta_P(t) -\sqrt{ |v|\psi_P(t)\tau_P}\,, \quad \text{and} \\
    \theta_{P,u}(t,v) &:= \theta_P(t) +  \sqrt{|v|\psi_P(t)\tau_P}\,. \label{eq:eff_bound}
\end{align*}
\end{prop}

Proposition~\ref{prop:par_id} is similar to Theorem~2 of \cite{chernozhukov2022long} with $I(T>t)$ taking the place of the outcome $Y$. Proposition~\ref{prop:par_id} yields upper and lower bounds for $\theta_c(t)$ that depend on $\theta_P(t)$, $\psi_P(t)$, and $\tau_P$, which are all mappings of the observed-data distribution, and the two sensitivity parameters $s_{c,T}(t)$ and $s_{c,A}$. Therefore, for fixed $(t,v)$, the upper and lower bounds are mappings of the observed-data distribution. To conduct a sensitivity analysis for the impact of unobserved confounding, we will use semiparametric efficiency theory to construct cross-fitted one-step estimators for $\theta_{P,l}(t,v)$ and $\theta_{P,u}(t,v)$ along with confidence intervals for fixed $(t,v)$ and uniform confidence regions over $(t,v)$. We discuss the details of estimation and inference in Section \ref{sec:estinf}. We also note that the upper and lower bounds follow by plugging in 1 as an upper bound for the correlation $|\rho_c(t)|$. This is likely often a conservative bound; we discuss a possible approach for obtaining a sharper bound in Section~\ref{sec:senspar2}.

The restricted mean survival time (RMST) is commonly used in survival analysis as a summary of the survival function over a given interval. In the context of causal inference, instead of or in addition to the difference in counterfactual survival curves, we could focus on the difference in counterfactual RMSTs defined as 
\[ \phi_{c}(t)  := \int_0^t \left[ P_{c}(T(1)>v) - P_{c}(T(0)>v) \right] \, dv = E_c \left[ \min\{ T(1), t \} - \min\{T(0), t\} \right]\,.\]
The effect bounds of Proposition~\ref{prop:par_id} can be extended to $\phi_{c}(t)$ by replacing $I(T >t)$ with $\min\{T, t\}$.  Sensitivity analysis for the difference in counterfactual RMSTs can then be conducted in a similar manner to the methods described in the remainder of the paper. Additional details are provided in Supplementary Material.

\subsection{Interpretation of the sensitivity parameters}
We now explain the interpretation of the sensitivity parameters $s_{c,T}(t)$ and $s_{c,A}$. Both parameters can be interpreted as certain nonparametric $R^2$ measures, which we define now. For a scalar random variable $V$, the \textit{nonparametric $R^2$} \citep{pearson1905general,doksum1995nonparametric,williamson2021nonparametric,chernozhukov2022long} is defined as $\eta_{V \sim W}^2:= \n{Var}(E[V\mid W]) / \n{Var}(V)\,.$ Similar to $R^2$ in a linear regression model (though here no specific model is assumed), $\eta_{V \sim W}^2$ is the fraction of variance of $V$ explained by the regression of $V$ onto $W$, ranging from 0 if $E(V \mid W)$ is  degenerate to 1 if $E(V\mid W) = V$. The \textit{nonparametric partial $R^2$} of $V$ on $U$ given $X$ is defined as 
$$\eta_{V \sim U\mid X}^2:=\frac{\n{Var}(E[V \mid X,U]) - \n{Var}(E[V \mid X)]}{\n{Var}(V) - \n{Var}(E[V \mid X])} = \frac{\eta_{V \sim XU}^2 - \eta_{V \sim X}^2}{1-\eta_{V \sim X}^2}$$ \citep{williamson2021nonparametric,chernozhukov2022long,williamson2023general}.
Analogously, $\eta_{V \sim U\mid X}^2$ is the fraction of the residual variance in $V$ explained by $U$ after adjusting for $X$, ranging from 0 if $E(V \mid X, U) = E(V \mid X)$ to 1 if $E(V \mid X, U) = V$. By the definitions of $g_{c,t}$ and $g_{P,t}$, we can write
\[ s_{c,T}(t) = \frac{\n{Var}(E[I(T>t)\mid A,W,U]) - \n{Var}(E[I(T>t) \mid A,W))]}{\n{Var}(I(T>t)) - \n{Var}(E[I(T>t) \mid A,W])} = \eta^2_{I(T>t) \sim  U\mid A,W}\,,\] 
so that $s_{c,T}(t)$ can be interpreted as the residual variance in $I(T > t)$ explained by $U$ after adjusting for $(A,W)$. Similarly, by the definitions of $\alpha_c$ and $\alpha_P$, we can write
\[ s_{c,A} = \frac{E \left[1 / \n{Var}(A \mid W, U)\right] - E \left[1 / \n{Var}(A \mid W)\right]}{E \left[1 / \n{Var}(A \mid W, U)\right]}\,,\] 
so that $s_{c,A}$ can be interpreted as the the proportion of the mean conditional \emph{precision} of $A$ given $(W, U)$ not explained by $W$ alone.

Therefore, unlike coefficients of unobserved confounders in parametric regression models, our sensitivity parameters have model-agnostic, nonparametric interpretations without assumptions about the dimension or distribution of the unobserved confounders or that they are independent of the observed covariates. Furthermore, both $s_{c,T}(t)$ and $s_{c,A}$ are independent of the variance of the covariates and are contained in $[0,1]$. This scale-free property allows us to \textit{benchmark} the magnitude of the sensitivity parameters against values from the observed covariates \citep{cinelli2020making,veitch2020sense,chernozhukov2022long}, which we will discuss further in Section~\ref{sec:senspar}.

\section{Estimation and inference}\label{sec:estinf}
\subsection{Estimation}\label{sec:estimation}
We now define estimators of the effect bounds $\theta_{P,l}(t, v)$ and $\theta_{P,u}(t,v)$ for fixed $(t,v)$. Our estimators will be given by $\theta_n(t) \pm \sqrt{ |v|\psi_n(t) \tau_n}$  for estimators $\theta_n(t)$, $\psi_n(t)$, and $\tau_n$ of $\theta_P(t)$, $\psi_P(t)$ and $\tau_P$, respectively. We will use the cross-fitted, one-step estimator $\theta_{n}(t)$ of $\theta_P(t)$ proposed in \cite{westling2023inference}. We denote the efficient influence function (EIF) of $\theta_P(t)$ as $D^\ast_{P,\theta,t}$. We refer the reader to \cite{westling2023inference} for the definition of $D^\ast_{P,\theta,t}$ and details of the construction of $\theta_{n}$.

To estimate $\psi_P(t)$ and $\tau_P$, we will make use of the fact that these parameters are \emph{pathwise differentiable} relative to a nonparametric model, meaning intuitively that they are smooth enough as a function of the data-generating distribution to permit $n^{-1/2}$-rate estimation. We refer the reader to \cite{kennedy2016semiparametric} and references therein for a review of semiparametric efficiency theory. We define $G_P(t \mid a, w):=P(C \geq t \mid A=a, W=w)$ as the left-continuous conditional survival function of the censoring time $C$ and $\Lambda_P\left(t \mid a, w\right)$ as the cumulative hazard function corresponding to $S_P(t \mid a, w)$. We present the nonparametric EIFs of $\psi_P(t)$ and $\tau_P$ in the following propositions.

\begin{prop}\label{prop:eif1}
If there exists $\kappa > 0$ such that $G_P(t \mid a, w) \geq \kappa$ for each $a \in \{0,1\}$ and $P$-almost every $w$ such that $S_P(t \mid a, w) > 0$, then $\psi_P(t)$ is pathwise differentiable in a nonparametric model with EIF $D_{P,\psi,t}^\ast := D_{P,\psi,t} - \psi_P(t)$ where $D_{P,\psi,t}(y, \delta, a,w)$ equals
\begin{align*}
    &\left[1-2S_P\left(t \mid a, w\right)\right] S_P\left(t \mid a, w\right)\left[-\frac{I(y \leq t, \delta=1)}{S_P\left(y \mid a, w\right) G_P\left(y \mid a, w\right)}+\int_0^{t \wedge y} \frac{\Lambda_P\left(d u \mid a, w\right)}{S_P\left(u \mid a, w\right) G_P\left(u \mid a, w\right)}\right]\\
     &\qquad + S_P\left(t \mid a, w\right)\left[1-S_P\left(t \mid a, w\right)\right]\,.
\end{align*}
\end{prop}

\begin{prop}\label{prop:eif2}
If there exists $\kappa > 0$ such that $\pi_P(w) \in [\kappa, 1-\kappa]$ for $P$-almost every $w$, then $\tau_P$ is a pathwise differentiable parameter in a nonparametric model with EIF $D_{\tau}^\ast := D_{\tau} - \tau_P$ where 
\begin{align*}
    D_{P,\tau}(a,w) &= \frac{2}{\pi_P(w)\left[1-\pi_P(w)\right]} - \frac{\left[a -\pi_P(w)\right]^2}{\pi_P(w)^2\left[1-\pi_P(w)\right]^2}\,.
\end{align*}
\end{prop}

Based on these EIFs, we construct cross-fitted one-step estimators $\psi_n(t)$ and $\tau_n$ of $\psi_P(t)$ and $\tau_P$, respectively. The EIFs depend on the nuisance functions $S_{P}$, $G_{P}$ and $\pi_{P}$. In order to avoid Donsker conditions on our nuisance function estimators, we employ cross-fitting. We randomly split the observations into $K$ disjoint groups $\{\mathcal{V}_{n, 1},\dotsc,\mathcal{V}_{n, K}\}$ of approximately equal size. For each $k \in \{1, \dotsc, K\}$, we  construct nuisance estimators $S_{n,k}$, $G_{n,k}$ and $\pi_{n,k}$ using the training set $\mathcal{T}_{n, k}:=\left\{O_i: i \notin \mathcal{V}_{n, k}\right\}$. There is a one-to-one relationship between $S_P$ and $\Lambda_{P}$, so $\Lambda_{n,k}$ can be obtained from $S_{n,k}$. We then construct $D_{n,\psi,k,t}$ and $D_{n,\tau,k}$ by substituting $S_{n,k}$, $G_{n,k}$, $\pi_{n,k}$, and $\Lambda_{n,k}$ for  their counterparts in $D_{P,\psi,t}$ and $D_{P,\tau}$. We then define our cross-fitted one-step estimators
\begin{align}
\psi_{n}(t) :=\frac{1}{n} \sum_{k=1}^K \sum_{i \in \mathcal{V}_{n, k}} D_{n,\psi,k,t}\left(O_i\right)\,, \ \text{and } \, 
\tau_{n}  :=\frac{1}{n} \sum_{k=1}^K \sum_{i \in \mathcal{V}_{n, k}} D_{n,\tau,k}\left(O_i\right)\,. \label{eq:est1}
\end{align}

Both $\psi_P(t)$ and $\tau_P$ are non-negative, but the one-step estimators $\psi_{n}(t)$ and $\tau_{n}$ can be negative. In order to provide estimators that are guaranteed to be non-negative, we can use the one-step estimator if it is positive and otherwise use the plug-in estimator:
\begin{align*}
\psi_n^+(t) & :=I\left\{\psi_{n}(t)>0\right\}\psi_{n}(t)+I\left\{\psi_{n}(t) \leq 0\right\}  \frac{1}{n} \sum_{k=1}^K \sum_{i \in \mathcal{V}_{n, k}} S_{n,k}(t \mid A_i, W_i) \left[ 1-S_{n,k}(t \mid A_i, W_i)\right] \quad \text{and}\\
\tau_n^+ & :=I\left\{\tau_{n}>0\right\} \tau_{n}+I\left\{\tau_{n} \leq 0\right\} \frac{1}{n} \sum_{k=1}^K \sum_{i \in \mathcal{V}_{n, k}}\frac{1}{\pi_{n,k}(W_i)\left[1-\pi_{n,k}(W_i)\right]}\,.
\end{align*}
If $\psi_P(t) >0$ and $\tau_P > 0$, then the asymptotic properties of $\psi_n^+(t)$ and $\tau_n^+$ are the same as those of $\psi_n(t)$ and $\tau_n$.

For fixed $(v,t)$, we define our estimators of the effect bounds as
\begin{equation}
\begin{aligned}
    \theta_{n,l}(t,v) & = \theta_{n}(t) - \sqrt{|v|\psi_n^+(t)\tau_n^+}\,, \quad \text{and} \\
    \theta_{n,u}(t,v) & = \theta_{n}(t) + \sqrt{|v|\psi_n^+(t)\tau_n^+}\,.
\end{aligned}
\end{equation}
By the delta method, the influence functions of $\theta_{n,l}(t,v)$ and $\theta_{n,u}(t,v)$ are given by
\begin{equation}
\begin{aligned}
     D^\ast_{P,l, t,v} &= D^\ast_{P,\theta,t} - \frac{1}{2}\sqrt{\frac{|v|}{\psi_P(t)\tau_P}}\left[ \tau_P D^\ast_{P,\psi,t} + \psi_P(t) D^\ast_{P,\tau} \right]\,, \quad \text{and}\\
     D^\ast_{P,u, t,v} &= D^\ast_{P,\theta,t} + \frac{1}{2}\sqrt{\frac{|v|}{\psi_P(t)\tau_P}} \left[ \tau_P D^\ast_{P,\psi,t} + \psi_P(t) D^\ast_{P,\tau}\right]\,.
\end{aligned}
\end{equation}

We now discuss the large-sample properties of our proposed estimators $\theta_{n,l}(t,v)$ and $\theta_{n,u}(t,v)$. We let $G_\infty$ be the limit of $G_P$, which for our consistency result does not necessarily need to equal $G_P$. We then define
\begin{align*}
r_{n, 1} & :=\max_k E_P\left[\pi_{n, k}(W)-\pi_{P}(W)\right]^2\,;\\
r_{n, t, 2} & :=\max_k E_P\left[\sup _{u \in[0, t]}\left|G_{n, k}(u \mid A, W)-G_{\infty}(u \mid A, W)\right|^2\right]\,;  \\
r_{n, t, 3} & :=\max _k E_P\left[\sup _{u \in[0, t]}\left|\frac{S_{n, k}(t \mid A, W)}{S_{n, k}(u \mid A, W)}-\frac{S_{P}(t \mid A, W)}{S_{P}(u \mid A, W)}\right|^2\right] \text{ and}\\
r_{n,t,4} &:= \max _k E_P\left[\sup _{u \in[0, t]} \sup _{v \in[0, u]} \left|\frac{S_{n, k}(u \mid A, W)}{S_{n, k}(v \mid A, W)}-\frac{S_{P}(u \mid A, W)}{S_{P}(v \mid A, W)} \right|^2\right]\,.
\end{align*}
We then have the following conditions for consistency of our estimators.
\begin{description}[style=multiline,leftmargin=2cm, labelindent=.9cm]
    \item[\namedlabel{itm:nuisance}{(B1)}] There exists $G_{\infty}$ such that $r_{n, 1}$, $r_{n,t,2}$, and $r_{n,t,3}$ are all $\fasterthan(1)$.
    \item[\namedlabel{itm:positivity}{(B2)}] There exists $\eta>0$ such that, with probability tending to one, for $P$-almost all $w, \pi_{n, k}(w) \geq 1 / \eta$, $\pi_{P}(w) \geq 1 / \eta, G_{n, k}(t \mid a, w) \geq 1 / \eta$, and $G_{\infty}(t \mid a, w)$ $\geq 1 / \eta$.
    \item[\namedlabel{itm:uniform}{(B3)}] It holds that $r_{n,t,4} = \fasterthan(1)$. 
\end{description}

\begin{thm}[Consistency]
\label{thm:consistency}
If conditions~\ref{itm:nuisance}--\ref{itm:positivity} hold, then $\theta_{n}(t) \inprob \theta_{P}(t)$, $\psi_{n}(t) \inprob \psi_{P}(t)$, and $\tau_{n} \inprob \tau_{P}$. Then by the continuous mapping theorem, 
$$\theta_{n,l}(t,v) \inprob \theta_{P,l}(t,v) \quad \text{and} \quad \theta_{n,u}(t,v) \inprob \theta_{P,u}(t,v)\,.$$ 
If condition~\ref{itm:uniform} also holds, then $\sup_{u \in [0,t]} \left| \theta_{n}(u) - \theta_{P}(u)\right| \inprob 0$ and $\sup_{u \in [0,t]} \left| \psi_{n}(u) - \psi_{P}(u)\right| \inprob 0$, so
\[\sup _{u \in[0, t]}\left|\theta_{n,l}(t,v)-\theta_{P,l}(t,v)\right| \inprob 0 \quad \text{and} \quad \sup _{u \in[0, t]}\left|\theta_{n,u}(t,v)-\theta_{P,u}(t,v)\right| \inprob 0\,.\]
\end{thm}

Theorem~\ref{thm:consistency} provides conditions under which $\theta_{n,l}(t,v)$ and $\theta_{n,u}(t,v)$ are (uniformly) consistent. Condition~\ref{itm:nuisance} requires that $S_{n,k}$ and $\pi_n$ are consistent, but allows $G_{n,k}$ to be inconsistent. Therefore, $\theta_{n,l}(t,v)$ and $\theta_{n,u}(t,v)$ are robust to estimation of the survival function of censoring. Condition~\ref{itm:positivity} requires positivity of treatment assignment and censoring, and  condition~\ref{itm:uniform} requires a stronger type of convergence of $S_{n,k}$ to $S_P$ for uniform consistency.

We now introduce additional conditions under which $\theta_{n,l}(t,v)$ and $\theta_{n,u}(t,v)$ are asymptotically linear. 
\begin{description}[style=multiline,leftmargin=2cm, labelindent=.9cm]
\item[\namedlabel{itm:asy}{(B4)}] It holds that $G_\infty = G_P$ and $r_{n, 1}$, $r_{n,t,2}$, and $r_{n,t,3}$ are $\fasterthan(n^{-1/2})$.
\item[\namedlabel{itm:asy_unif}{(B5)}] It holds that $r_{n,t,4} = \fasterthan(n^{-1/2})$.
\end{description}
We define $\d{P}_n$ as the empirical distribution of the observed data.

\begin{thm}[Asymptotic linearity]
\label{thm:asy_linear}
If conditions~\ref{itm:nuisance}--\ref{itm:asy} hold with  $G_{\infty}=G_P$, $\tau_P >0$, and $\psi_P(t) > 0$, then $\theta_{n}(t) - \theta_{P}(t) = \d{P}_n D^\ast_{P,\theta,t} + \fasterthan(n^{-1/2})$, $\psi_{n}(t) - \psi_{P}(t) = \d{P}_n D^\ast_{P,\psi,t} + \fasterthan(n^{-1/2})$, $\tau_{n} - \tau_{P} = \d{P}_n D^\ast_{P,\tau} + \fasterthan(n^{-1/2})$, $\theta_{n,l}(t,v) - \theta_{P,l}(t,v) = \d{P}_n D^\ast_{P,l, t,v} + \fasterthan(n^{-1/2})$, and $\theta_{n,u}(t,v) - \theta_{P,u}(t,v) = \d{P}_n D^\ast_{P,u, t,v} + \fasterthan(n^{-1/2})$. If in addition condition~\ref{itm:asy_unif} holds and $\inf_{s \in [0,t]} \psi_P(s) > 0$, then
\begin{align*}
    & \sup_{v \leq M} \sup _{s \in[0, t]}\left|\theta_{n,l}(s,v) - \theta_{P,l}(s,v)-\d{P}_n D^\ast_{P,l, s,v}\right|=\fasterthan(n^{-1/2}) \quad \text{and} \quad \\
    & \sup_{v \leq M}\sup _{s \in[0, t]}\left|\theta_{n,u}(s,v) - \theta_{P,u}(s,v)-\d{P}_n D^\ast_{P,u,s,v}\right|=\fasterthan(n^{-1/2})\,.
\end{align*}
for any $M < \infty$.
\end{thm}

Theorem~\ref{thm:asy_linear} provides  conditions under which $\theta_{n,l}(t,v)$ and $\theta_{n,u}(t,v)$ are (uniformly) asymptotically linear. Condition~\ref{itm:asy} requires that the rates of convergence of $\pi_{n,k}$, $S_{n,k}$, and $G_{n,k}$ are faster than $n^{-1/4}$, and condition~\ref{itm:asy_unif} requires a stronger type of convergence of $S_{n,k}$ for uniform asymptotic linearity. These rates of convergence can be achieved under correctly specified Cox proportional hazard models for the event and censoring survival functions and a correctly-specified logistic regression models for the propensity score. These are reasonable choices in small samples. However, in larger samples, the risk of model misspecification can be reduced by using nonparametric and semiparametric methods. Ensemble estimators can then be used to choose weights for a set of candidate parametric, semiparametric, and nonparametric methods \citep{ishwaran2004relative, hothorn2006survival, van2007super, westling2023inference}.

\subsection{Pointwise inference}\label{sec:inference}
Asymptotic linearity of the effect bound estimators implies that $n^{1/2} [\theta_{n,l}(t,v) - \theta_{P,l}(t,v)]$ and $n^{1/2} [\theta_{n,u}(t,v) - \theta_{P,u}(t,v)]$ converge jointly in distribution to a mean-zero bivariate normal distribution with variances $\sigma_{P,l,t,v}^2 := P (D^\ast_{P,l, t,v})^2$ and $\sigma_{P,u,t,v}^2 := P (D^\ast_{P,u, t,v})^2$, respectively, and covariance $\Sigma_{P, ul, t, v} := P (D^\ast_{P,u, t,v}D^\ast_{P,l, t,v})$. We define the cross-fitted variance estimator $\sigma_{n,l,t,v}^2 := \frac{1}{n} \sum_{k=1}^K \sum_{i \in \mathcal{V}_{n, k}}[D^\ast_{n,l,t,v}\left(O_i\right)]^2$, and we analogously define $\sigma_{n,u,t,v}^2$ and $\Sigma_{n, ul, t, v}$. We then define an asymptotic $(1-\alpha)$-level Wald-type confidence interval as
\begin{align}
    \left[\ell_n(t,v),\; u_n(t,v)\right] = \left[\theta_{n,l}(t,v) - n^{-1/2}c_{n,t, v,\alpha}, \; \theta_{n,u}(t,v) + n^{-1/2}c_{n,t, v,\alpha}\right]\,,
\end{align}
where $c_{n, t, v, \alpha}$ is such that $P(Z_1 \leq c_{n, t, v, \alpha}, Z_2 \geq -c_{n, t, v, \alpha}) =(1-\alpha)$, where $(Z_1, Z_2)$ follow a mean-zero bivariate normal distribution with cross-fitted estimated covariance as above.
Under the conditions of Theorem~\ref{thm:asy_linear}, $P( \ell_n(t,v) \leq \theta_{P,l}(t,v), \theta_{P,u}(t,v) \leq u_n(t,v))$ converges to $1-\alpha$. Therefore, by Proposition~\ref{prop:par_id}, if the product of the sensitivity parameters is at most $v$, i.e., $s_{c,T}(t) s_{c,A}/(1-s_{c,A}) \leq v$, then $P(\ell_n(t,v) \leq \theta_c(t) \leq  u_n(t,v))$ converges to at least $1-\alpha$. Since $\theta_{P,l}(t) \in [-1,1]$ and $\theta_{P,u}(t) \in [-1,1]$, we can alternatively construct confidence intervals that respect these bounds using a log transformation \citep{anderson1982approximate}---see Appendix~\ref{app:transformation} of Supplementary Material for details.  We also note that the interval $\left[\theta_{n,l}(t,v) - z_{1-\alpha/2}n^{-1/2}\sigma_{n,l,t,v}, \; \theta_{n,u}(t,v) + z_{1-\alpha/2}n^{-1/2}\sigma_{n,u,t,v}\right]$, which does not make use of the covariance between the upper and lower effect bound estimators, is asymptotically conservative. Here, $z_p$ is the $p$th quantile of a standard normal distribution.

\subsection{Uniform inference}\label{sec:uniform}
We now propose uniform confidence bands and hypothesis tests for the causal survival difference $\theta_c$ under unobserved confounding bounded by $v$. Since $\lim_{t \rightarrow 0} \psi_{P}(t)=\lim _{t \rightarrow t^{+}} \psi_{P}(t)=0$ for $t^{+}=\inf \{t > 0:\psi_{P}(t)=0\}$, we note that the uniform asymptotic linearity of $\theta_{n,l}(t,v)$ and $\theta_{n,u}(t,v)$ are invalid near $t=0$ and $t=t^{+}$. Therefore, we construct confidence bands and uniform tests on intervals of the form $[t_0,t_1]$ for $t_0>0$ and $t_1<t^{+}$. We define $\d{G}_{n,l,v}$ and $\d{G}_{n,u,v}$ as the processes $\{n^{1/2} [\theta_{n,l}(t,v) - \theta_{P,l}(t,v)]: t \in [t_0,t_1]\}$ and $\{n^{1/2} [\theta_{n,u}(t,v) - \theta_{P,u}(t,v)]: t \in [t_0,t_1]\}$, respectively. Uniform asymptotic linearity of the effect bound estimators implies that $\d{G}_{n,l,v}$ and $\d{G}_{n,u,v}$ converge jointly weakly to mean-zero correlated Gaussian processes $\xi_{l,v}$ and $\xi_{u,v}$ with $\n{Cov}(\xi_{l,v}(r), \xi_{l,v}(s)) = P(D^\ast_{P, l, r,v}D^\ast_{P, l, s,v})$, $\n{Cov}(\xi_{u,v}(r), \xi_{u,v}(s)) = P(D^\ast_{P, u, r,v}D^\ast_{P, u, s,v})$, and $\n{Cov}(\xi_{l,v}(r), \xi_{u,v}(s)) = P(D^\ast_{P, l, r,v}D^\ast_{P, u, s,v})$.  
A fixed-width asymptotic $(1-\alpha)$-level uniform confidence band is then given by
\begin{align}
    \left[\tilde\ell_n(t,v),\;\tilde u_n(t,v)\right] = \left[\theta_{n,l}(t,v) - n^{-1/2} q_{n,v,\alpha}, \; \theta_{n,u}(t,v) + n^{-1/2}q_{n,v,\alpha}\right]\,,
\end{align}
where $q_{n,v,\alpha}$ is such that $P(\sup_{t \in [t_0,t_1]} \xi_{n,l,t} \leq q_{n,v,\alpha}, \, \sup_{t \in [t_0,t_1]} \xi_{n,u,t} \geq -q_{n,v,\alpha}) =(1-\alpha)$ for $(\xi_{n,l,t}, \xi_{n,u,t})$ correlated mean-zero Gaussian processes with cross-fitted estimated covariances. By Proposition~\ref{prop:par_id}, if $s_{c,T}(t) s_{c,A}/(1-s_{c,A}) \leq v$ for all $t \in [t_0,t_1]$, then $\tilde\ell_n(t,v) \leq \theta_{c}(t) \leq \tilde{u}_n(t,v)$ \emph{for every} $t \in [t_0,t_1]$ with probability converging to $1-\alpha$. We can alternatively construct a variable-width confidence band that respects the bounds of $\theta_c$ using a log transformation and scaling \citep{anderson1982approximate,westling2023inference}---see Appendix~\ref{app:transformation} of Supplementary Material for details.

We can also use uniform asymptotic linearity to test the null hypothesis that $\theta_c(t) = \theta_0$ for all $t \in [t_0,t_1]$ and a fixed $\theta_0 \in [-1,1]$ under unobserved confounding bounded by $v$. Most often, researchers are interested in testing whether the effect bounds include the null effect, so that $\theta_0=0$.  By Proposition~\ref{prop:par_id}, this is equivalent to testing 
\begin{equation} \label{eq:uniform_testing}
 \begin{aligned}
    H_0: & \, \theta_{P,l}(t,v) \leq \theta_0 \leq \theta_{P,u}(t,v ) \; \text{ for all } t \in [t_0,t_1] \quad \text{ vs }  \\
    H_A: & \, \theta_{P,l}(t,v)>\theta_0 \text{ or } \theta_{P,u}(t,v)<\theta_0 \; \text{ for some } t \in [t_0,t_1] \,.
\end{aligned}
\end{equation}
This test reduces to assessing whether the uniform confidence band derived above contains a flat line; see Appendix~\ref{app:transformation} for details.

\section{Sensitivity analysis}\label{sec:senspar}
In this section, we use the theory and methods presented so far to conduct formal causal sensitivity analysis. We break down this process into two steps: (1) determine the minimum values of unobserved confounding needed to reverse the causal conclusion; (2) assess the plausibility of these minimum values. In Section \ref{sec:senspar1}, we use \textit{robustness values} \citep{cinelli2020making,chernozhukov2022long} to address step (1). In Section \ref{sec:senspar2}, we use \textit{benchmarking} \citep{cinelli2020making,chernozhukov2022long} to address step (2). 

\subsection{Robustness values}\label{sec:senspar1}
For a fixed effect size $\theta_0 \in [-1,1]$, which is most often $\theta_0 = 0$, we define the \textit{robustness value} at time $t$ as 
\begin{equation}\label{eq:rv}
    \mathrm{RV}_P(t, \theta_0):=\inf \left\{q \in[0,1): \theta_{P,l}(t, q^2/[1-q]) \leq \theta_0 \leq \theta_{P, u}(t,q^2/[1-q])\right\}\,.
\end{equation}
In words, RV$_P(t, \theta_0)$ is the smallest amount of unobserved confounding, i.e., the smallest $q$ with $s_{c,T}(t) = q$ and $s_{c,A} = q$, that makes $\theta_0$ contained within the effect bounds. Rearranging, $\theta_{P,l}(t, q^2/[1-q]) \leq \theta_0 \leq \theta_{P, u}(t,q^2/[1-q])$ if and only if $[\theta_P(t) - \theta_0]^2 / [\psi_P(t) \tau_P] \leq q^2 / [1-q]$. Hence, RV$_P(t, \theta_0)$ is given in closed form as the positive solution to the quadratic equation $q\mapsto q^2+\lambda_Pq-\lambda_P = 0$, where   $\lambda_P=[\theta_{P}(t) - \theta_0]^2/[\psi_P(t)\tau_P]$.  An asymptotically linear estimator of RV$_P(t, \theta_0)$ can thus be obtained by replacing $\lambda_P$ with $\lambda_n:=[\theta_{n}(t) - \theta_0]^2/[\psi_n^+(t)\tau_n^+]$ in the quadratic formula.

To account for uncertainty in estimation of $\theta_P$, $\tau_P$, and $\psi_P$, we define the \textit{minimum influential robustness value}, $\mathrm{MIRV}_{n,\alpha}(t,\theta_0)$, at time $t$ as 
\begin{equation}
\inf \left\{q \in[0,1): \text{fail to reject } H_0 : \theta_{P,l}(t,q^2/[1-q]) \leq \theta_0 \leq \theta_{P, u}(t,q^2/[1-q]) \text{ at level } \alpha \right\}\,.
\end{equation}
Thus defined, MIRV$_{n,\alpha}(t,\theta_0)$ represents the threshold of unobserved confounding at which one of the effect bounds shifts from statistically significant to statistically insignificant. Intuitively, a robustness value close to 0 suggests that a small amount of unobserved confounding could reverse the causal conclusion, indicating that the evidence of a causal effect is sensitive to unobserved confounding. By contrast, a robustness value close to 1 implies that the causal effect estimate can only be explained away by strong unobserved confounders. However, assessing the plausibility of these robustness values remains challenging because ``close" is an ambiguous term that depends on the setting, as we will discuss in Section \ref{sec:senspar2}. We also note that the null hypothesis $H_0: \theta_{P,l}(t,q^2/[1-q]) \leq \theta_0 \leq \theta_{P, u}(t,q^2/[1-q])$ is equivalent to $H_0: \mathrm{RV}_P(t, \theta_0) \leq q $ so MIRV$_{n,\alpha}(t,\theta_0)$ can also be interpreted as a $(1-\alpha)$-level lower confidence limit for RV$_P(t, \theta_0)$.

Finally, to summarize the robustness of the causal conclusion uniformly over $t \in [t_0,t_1]$, we define the \textit{uniform minimum influential robustness value}, $\n{UMIRV}_{n,\alpha}(\theta_0)$, as the smallest $q$ at  which we fail to reject the null hypothesis that $\theta_{P,l}(t,q^2/[1-q]) \leq \theta_0 \leq \theta_{P,u}(t,q^2/[1-q]) \text{ for all } t \in [t_0,t_1]$ at level $\alpha$. This provides a single metric that summarizes the overall robustness of the survival difference to unobserved confounding over a time range. We note that robustness values and minimum influential robustness values can also be defined allowing $s_{c,T}(t)$ and $s_{c,A}$ to vary independently or when one of $s_{c,T}(t)$ or $s_{c,A}$ is known.

\subsection{Benchmarking}\label{sec:senspar2}
To determine whether a robustness value is reasonable or to provide a justifiable range for sensitivity parameters is challenging. The plausibility of the robustness value depends on the strength of unobserved confounders, about which we may not have external knowledge. In addition, even if some knowledge about unobserved confounding is available from prior studies or external data, incorporating data from various models for relative comparison may be difficult. A natural way to gain some insight into the plausibility of sensitivity parameters is to use observed confounding across a subset of the observed covariates, which is known as \textit{benchmarking} \citep{cinelli2020making,chernozhukov2022long} or \textit{calibration} \citep{mcclean2024calibrated}. 

We now provide details about the application of benchmarking to our setting. For any  $R \subseteq\{1,2 \ldots, p\}$, we define $W_R$ as the subvector of $W$ with indices in $R$, and $W_{-R}$ as the subvector excluding indices in $R$. The observed confounding by $W_{R}$ is then defined as
\begin{equation} \label{eq:obs_confounding}
    \begin{aligned}
     s_{P,T}(t,R)  &:= \frac{\n{Var}(g_{P,t}) - \n{Var}(g_{P,-R,t})}{\n{Var}(I(T>t)) -  \n{Var}(g_{P,t})} =\frac{E[g_{P,t}-g_{P,-R,t}]^2}{\psi_P(t)}  \quad \text{and}  \\
     s_{P,A}(R) &:=   1 - \frac{E\alpha_{P,-R}^2}{E\alpha_{P}^2}\,,
\end{aligned}
\end{equation}
where $\alpha_{P,-R} := \frac{a}{\pi_P(w_{-R})} - \frac{1-a}{1-\pi_P(w_{-R})}$ and $g_{P,-R,t}:=S(t \mid A,W_{-R})$. We use $s_{P,T}(t, R)$ and $s_{P,A}(R)$ as model-agnostic benchmarks for the sensitivity parameters $s_{c,T}(t)$ and $s_{c,A}$.  Here, $\n{Var}(I(T>t)) -  \n{Var}(g_{P,t})$ rather than $\n{Var}(I(T>t)) -  \n{Var}(g_{P,-R,t})$ appears in the denominator of $s_{P,T}(t, R)$ in order to improve comparability with $s_{c,T}(t)$. We then define the \textit{observed confounding by $W
_R$} as $s_P(t, R):=s_{P,T}(t, R) s_{P,A}(R) / [1-2s_{P,A}(R)]$. 

Since $s_{P,T}(t, R)$ and $s_{P,A}$ are functions of the observed-data distribution, they can be estimated using the observed data. Cross-fitted one-step estimator can be constructed following the methods in Section \ref{sec:estimation}. We suggest using plug-in estimators for simplicity and to ensure the estimators are non-negative. A plug-in estimator for $E[g_{P,t}-g_{P,-R,t}]^2$ is given by $\frac{1}{n} \sum_{k=1}^K \sum_{i \in \mathcal{V}_{n, k}}  [S_{n,k}(t) - S_{n,k,-R}(t)]^2$, where $ S_{n,k,-R}$ is a cross-fitted estimator of the conditional survival functions of $T$ given $A$ and $W_{R}$. A plug-in estimator for $E[\alpha_{P}-\alpha_{P,-R}]^2$ can be obtained analogously. These two estimators can then be used to form an estimator $s_{n}(t, R)$ of $s_{P}(t,R)$.

Finally, we explain how to assess the plausibility of the robustness value using $s_{n}(t, R)$. First, if there is a specific set of observed covariates $R$ that is expected to possess a similar level confounding as unobserved confounders, then we compare $r_n^2/(1-r_n)$ with $s_{n}(t,R)$, where $r_n$ can be RV$_n(t, \theta_0)$ or MIRV$_{n,\alpha}(t,\theta_0)$. If $s_{n}(t, R) < \n{RV}_{n}(t, \theta_0)^2 / [ 1 - \n{RV}_{n}(t, \theta_0)]$, then unobserved confounding of the same strength as the covariates in $R$ would result in $\theta_0$ falling outside the effect bounds at time $t$. If $s_{n}(t,R) < \n{MIRV}_{n,\alpha}(t,\theta_0)^2 / [1 -\n{MIRV}_{n,\alpha}(t,\theta_0)]$, then unobserved confounding of the same strength as the covariates in $R$ would still result in rejecting $H_0: \theta_c(t) = \theta_0$. Second, if knowledge about a specific set of observed covariates is unavailable, then one possibility is to define $d_P(t)$ as the maximal value of $d \in \{1,\dotsc, p\}$ such that robustness value is greater than or equal to the mean of $s_{P}(t,R)$ over all subsets $R$ of size $|R| = d$. If the number of subsets of size $d$ is too large to compute $s_{n}(t,R)$ for all $R$ with $|R| = d$, we can estimate the average using a random selection of subsets \citep{bonvini2022sensitivity}. We refer to the average observed confounding across a set of $d$ variables as the \textit{average leave-d-out} confounding. We illustrate the process of conducting a sensitivity analysis based on the robustness value and benchmarking in Section \ref{sec:data}.

We can also benchmark the correlation $\rho_c(t)$ using observed data in a similar way to find a plausible range of values for $\rho_c(t)$ in the effect bounds. The observed correlation for $R$ is given by
$$\rho_{P,-R}(t) := \n{Cor}\left(g_{P,t}-g_{P,-R,t}, \alpha_P-\alpha_{P,-R}\right) = \frac{\theta_{P}(t) - \theta_{P,-R}(t)}{\sqrt{E\left(g_{P,t}-g_{P,-R,t}\right)^2}  \sqrt{E\left(\alpha_P-\alpha_{P,-R}\right)^2}},$$
where $\theta_{P,-R}(t)$ is the observed data effect conditioning on $W_{-R}$ \citep{chernozhukov2022long}. The plug-in estimators for $E(g_{P,t}-g_{P,-R,t})^2$ and $E(\alpha_{P}-\alpha_{P,-R})^2$ and the cross-fitted one-step estimators for $\theta_{P,-R}(t)$ and $\theta_P(t)$ can be used to form an estimator $\rho_{n,-R}(t)$ of $\rho_{P,-R}(t)$. A value of $\rho_{n,-R}(t)$ smaller than 1 yields tighter effect bounds.

\section{Simulation studies}\label{sec:simulation}
We conducted a simulation study to examine the finite-sample properties of our methods. The data simulation process contained multiple continuous unobserved confounders correlated with observed confounders to illustrate the flexibility of our approach. Specifically, we generated independent unobserved confounders $U_1\sim \text{Uniform}(0,1)$ and $U_2 \sim \text{Uniform}(-2,2)$. Given $(U_1, U_2)$, we then simulated $W_1 \sim \text{Beta}(2U_1, 1)$ and $W_2 \sim \text{Uniform}(0,1)$. Given $U$ and $W$, we simulated $A$ from a Bernoulli distribution with probability $\pi_P(U,W):=\text{expit}(0.2-0.2W_1+0.1W_2-0.55U_1-0.5U_2)$. Given $A$ and $W$, we simulated the censoring time $C$ from an exponential distribution with rate $\lambda_{P,C}(A,W) :=\text{exp}(-0.5-0.15A-0.3W_1+0.1W_2)$. Given $A$, $U$ and $W$, we simulated the event time $T$ from an exponential distribution with rate $\lambda_{P,T}(A, U, W) := \text{exp}(0.15-0.25A-0.1\sqrt{W_1}-0.2W_2+0.5\sqrt{U_1}+1.75\text{exp}\{-3+U_2/2\})$. The average censoring rate at time $t=2$ was $E[P(C \leq 2\mid A,W)]=0.6$. We determined the true sensitivity parameters $s_{c,T}(t)$ and $s_{c,A}$ using numerical integration via the \texttt{cubature} package in \texttt{R} \citep{R-cubature}.

For each sample size $n \in \{500, 1000, 2500, 5000\}$, we simulated 1000 datasets using the process described above. For each dataset, we estimated the bounds $\theta_{n,l}(t)$ and $\theta_{n,u}(t)$ for $t \in \{0.5, 1, 1.5, 2\}$, using our proposed methods implemented in \texttt{R}. We estimated the conditional survival functions using generalized additive Cox regression models \citep{hastie1990generalized,klein2003survival} and the propensity score using a generalized additive logistic regression model \citep{hastie1990generalized}.

We now discuss the results of the numerical study. The first row of Figure~\ref{fig:bias} shows $\sqrt{n}$ times the bias of $\theta_{n,l}(t)$ and $\theta_{n,u}(t)$ for $t \in \{0.5, 1, 1.5, 2\}$ as a function of $n$. In general, the empirical bias was not significantly different from zero accounting for Monte Carlo error. The bias was slightly above zero for large $n$ at $t=0.5$. These empirical results illustrate that the bias of the proposed estimators tends to zero faster than $n^{-1/2}$. The second row of Figure~\ref{fig:bias} shows $n$ times the MSE of $\theta_{n,l}(t)$ and $\theta_{n,u}(t)$. The MSE appears to be proportional to $n$ as $n$ increases, consistent with the theoretical $n^{-1/2}$ rates of convergence of the estimators. The $n^{-1/2}$ rates of convergence of $\theta_{n,l}$ and $\theta_{n,u}$ demonstrate a primary advantage of EIF-based estimators: they can still achieve a rate of convergence of $n^{-1/2}$, even if the nuisances estimators converge slower than $n^{-1/2}$. 

\begin{figure}[h]
    \centering
    \includegraphics[width=\textwidth]{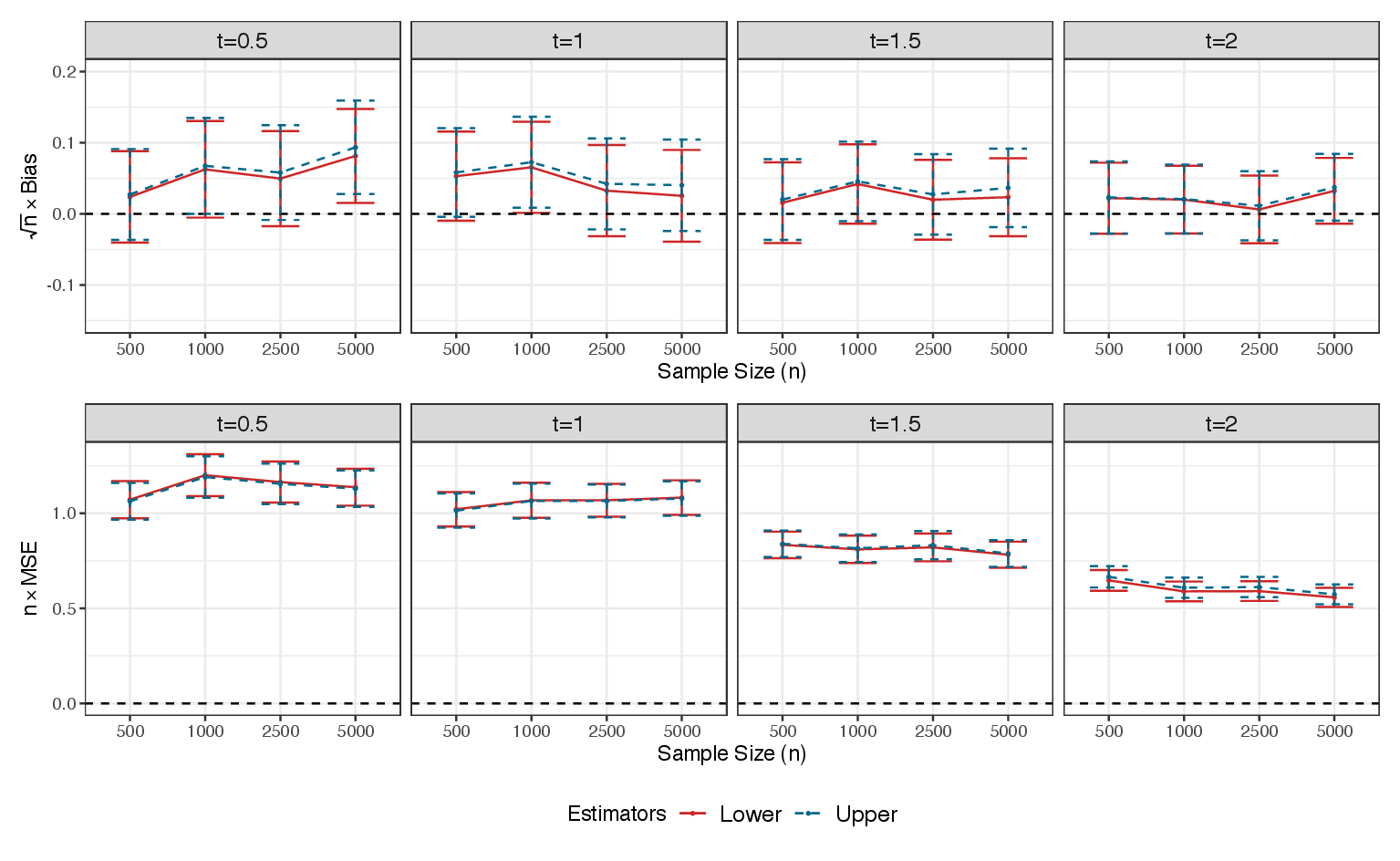}
    \caption{\label{fig:bias} Bias scaled by $\sqrt{n}$ (top) and MSE scaled by $n$ (bottom) of the estimators for the lower (red) and upper (blue) bounds of true causal effects as a function of $n$. Error bars indicate 95\% confidence intervals accounting for Monte Carlo error.}
\end{figure}

The top row of Figure~\ref{fig:ci} displays the empirical coverage of 95\% confidence pointwise intervals for the effect bounds. Both the standard Wald-type pointwise confidence intervals and transformed Wald-type intervals had coverage rates within Monte Carlo error of the nominal level for all time points and sample sizes. The bottom row of Figure~\ref{fig:ci} displays the empirical coverage of 95\% uniform confidence bands for the effect bounds over the interval $[0.1,2]$. The empirical coverage of the equi-width bands was generally within Monte Carlo error of the nominal level. The transformed bands had  empirical coverage rates that were slightly higher than the nominal level. 

\begin{figure}[h]
    \centering
    \includegraphics[width=\textwidth]{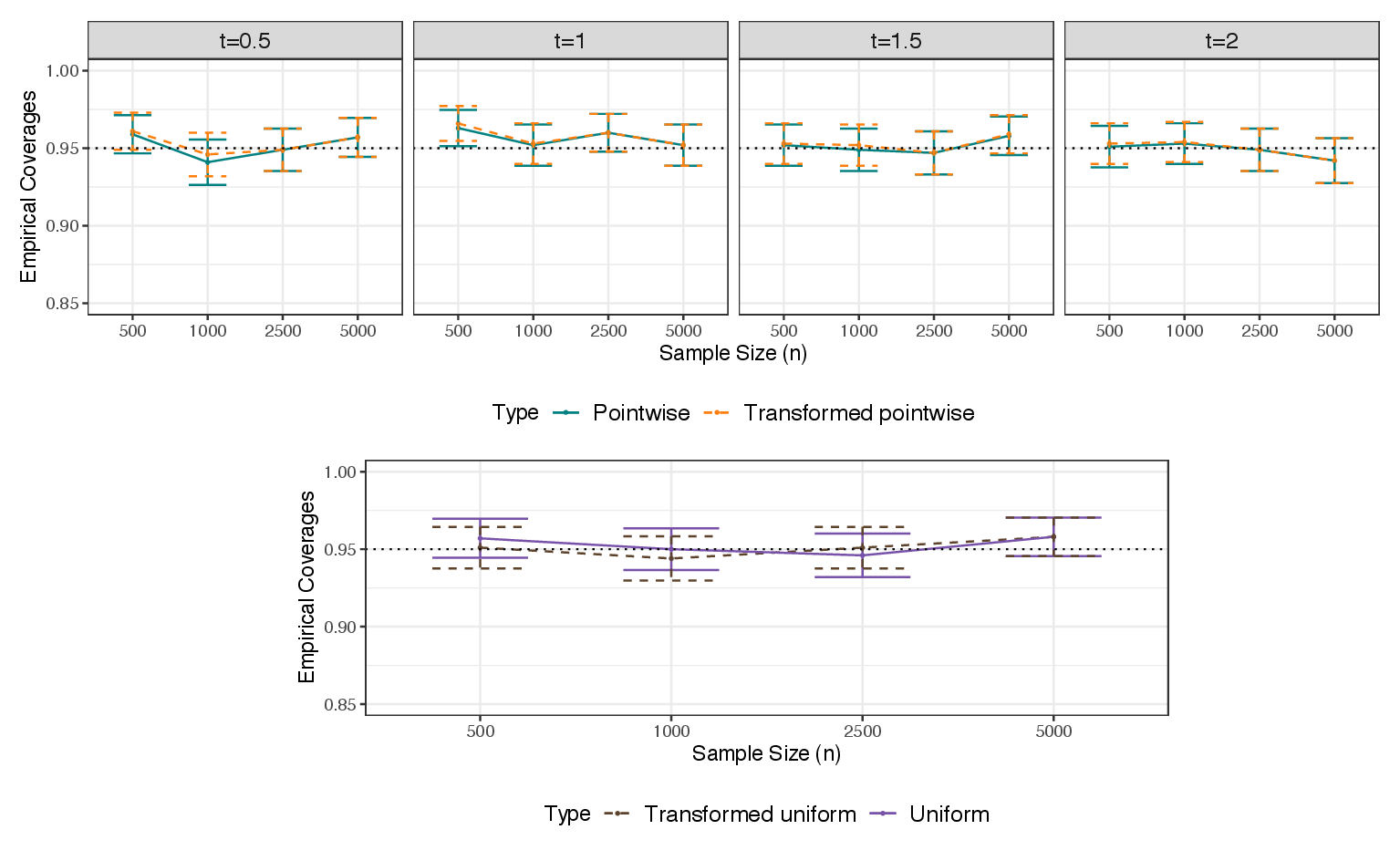}
    \caption{\label{fig:ci} Empirical coverage of 95\% pointwise confidence intervals (top) and uniform confidence bands (bottom) with Monte Carlo error bars as a function of $n$. Dashed lines represent transformed intervals and bands.}
\end{figure}

\section{Sensitivity analysis for the effect of Elective Neck Dissection on mortality}\label{sec:data}
In this section, we use our proposed methods to assess the evidence of a causal effect of elective neck dissection (END) on survival among patients with clinically node-negative, high-grade parotid carcinoma using data from a retrospective cohort study. The data consists of $n=1547$ patients who were diagnosed with clinically node-negative, high-grade parotid cancer between January 1, 2004 and December 31, 2013, and followed until the latter date. The exposure is receipt of END at diagnosis, denoted by $A=1$. The outcome, which is subject to right-censoring, was all-cause mortality up to five years post-diagnosis. Observed baseline confounders include age, sex, race, surgery status, tumor stage, histology, comorbidity, and payor, as well as the average income, education, county of residence, and treatment facility type.

We first estimated the effect bounds for each month during the five years following diagnosis under different levels of unobserved confounding  and constructed 95\% transformed pointwise intervals and uniform confidence bands. Next, we calculated RV$_n(t, 0)$ and MIRV$_{n,.05}(t, 0)$ for $t$ equal to one year post-diagnosis, as well as UMIRV$_{n,.05}(0)$ over the one-year post-diagnosis period. Finally, we assessed the plausibility of these robustness values by benchmarking against observed confounding. We used SuperLearner \citep{van2007super} to estimate the propensity score with the same candidate library as \citet{westling2023inference}: generalized linear models, generalized additive models, multivariate adaptive regression splines, random forests, and extreme gradient boosting. We used survSuperLearner \citep{westling2023inference} with a library consisting of the treatment group-specific Kaplan-Meier estimators, parametric survival models and generalized additive Cox proportional hazards models to estimate the conditional survival and censoring functions.

\begin{figure}[h!]
    \centering
    \includegraphics[width=\textwidth]{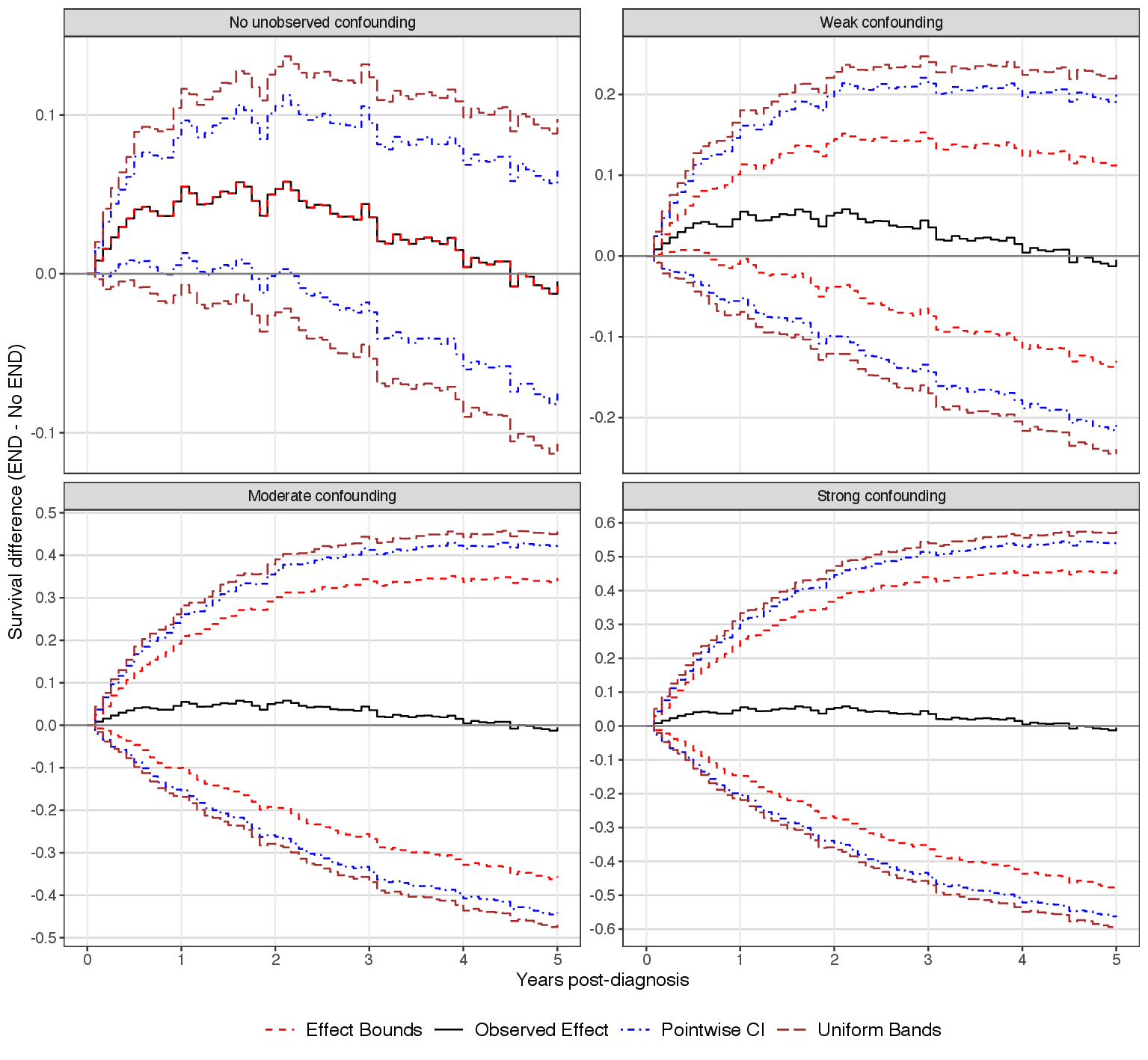}
    \caption{\label{fig:estEND} Estimated observed survival difference, estimated effect bounds, 95\% pointwise confidence intervals for the bounds, and 95\% uniform confidence bands for the bounds under different levels of unobserved confounding. Note the different y-axis scales in the four panels.}
\end{figure}

The top left panel of Figure~\ref{fig:estEND} displays the estimated counterfactual survival difference under no unobserved confounding along with 95\% pointwise confidence intervals and uniform confidence bands. We estimate the survival difference to be 5.5\% (95\% CI: 1.3\%--9.7\%) at one year post-diagnosis. The $p$-value of the test of the null hypothesis that $\theta_P(t)=0$ for all $t \in [1/12,1]$ was 0.015. Therefore, if there are no unobserved confounders then there is statistically significant evidence of a positive causal effect of END on short-term survival. The other three panels of Figure~\ref{fig:estEND} display the estimated bounds under weak, moderate and strong unobserved confounding. These unobserved confounding levels were determined by the average confounding when dropping 3, 8, and 13 of the 16 baseline confounders, as described in Section~\ref{sec:senspar2}.   Under all three levels of unobserved confounding, the pointwise 95\% confidence interval included the null effect for all time points between $1/12$ and 1 year post-diagnosis. This result suggests that even weak unobserved confounding could explain away the positive observed effect during the early post-diagnosis period.

Next, we present the results of the sensitivity analysis using the robustness values. We estimate the robustness value at one year post-diagnosis to be RV$_n(1, 0) = 0.032$, which suggests that to have the effect bounds include the null effect at $t=1$, it would have to hold that $s_{c,T}(1)s_{c,A}/(1-s_{c,A}) \geq (0.032)^2/(1-0.032)=1.06 \times 10^{-3}$. The minimum influential robustness value is MIRV$_{n,.05}(1, 0) = 0.008$, which indicates that to shift the statistically significant lower effect bound to statistically insignificant, the corresponding confounding metrics would need to be at least $6.45 \times 10^{-5}$. The uniform minimum influential robustness value over $[1/12, 1]$ is URV$_{n,.05}(0) = 0.006$, which indicates that to shift the statistically significant uniform effect bound to statistically insignificant, the corresponding confounding metrics would need to be at least $3.62 \times 10^{-5}$ for each $t \in [1/12, 1]$. 

We now use benchmarking against the observed covariates to assess the plausibility of the estimated robustness values as described in Section~\ref{sec:senspar2}. First, we use the \textit{leave-one-out} approach where we compare the robustness values with the observed confounding by each individual observed covariate $W_j$ \citep{lu2023flexible}. We find that \textit{surgery}, a categorical variable indicating whether the patient received surgery to remove the tumor, chemotherapy, and/or radiation therapy,  is the only covariate whose estimated confounding level is larger than the estimated robustness value at time $t = 1$ year. However, every observed covariate has a larger estimated confounding level than the MIRV at $t = 1$, suggesting that unobserved confounding as strong as any of the individual observed covariates could make the effect at $t = 1$ statistically insignificant. Similarly, unobserved confounding as strong as any of the individual observed covariates could make the uniform test over $t \in [1/12,1]$ statistically insignificant.  

\begin{figure}[h!]
    \centering
    \includegraphics[width=\textwidth]{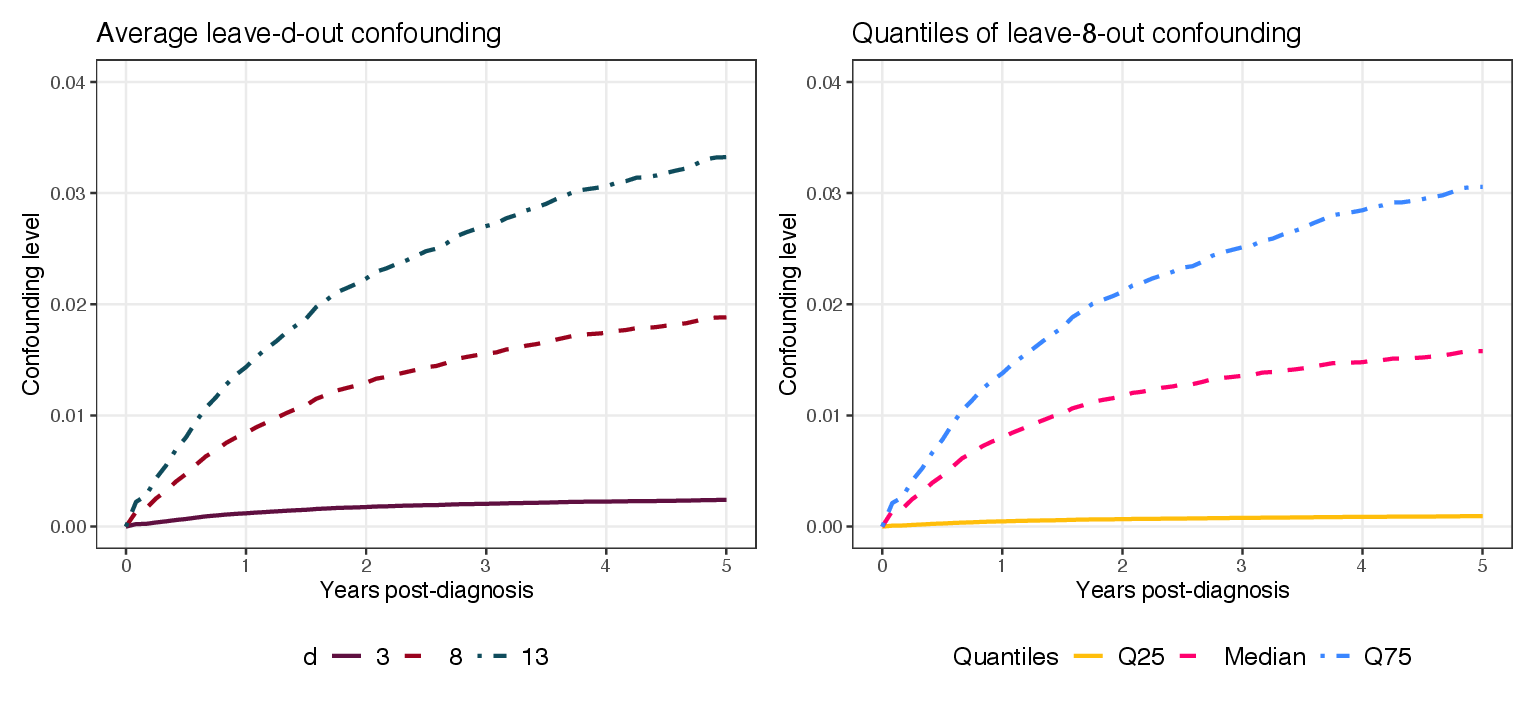}
    \caption{\label{fig:ENDsenspar2} Left: Average confounding levels when dropping $d$ covariates over time. Right: quartiles of the distribution of the confounding level when omitting $d=8$ covariates over time.}
\end{figure}

As discussed in Section~\ref{sec:senspar2}, we can alternatively compare the robustness values to the average estimated confounding when dropping $d$ covariates at a time. The left panel of Figure~\ref{fig:ENDsenspar2} displays the average \textit{leave-d-out} confounding for $d \in \{3,  8, 13\}$ over time. We find that the estimated robustness value at $t = 1$ fell between the average \textit{leave-2-out} confounding and \textit{leave-3-out} confounding. Therefore, unobserved confounding as strong as three or more randomly selected observed covariates could make the effect zero. Finally, the right panel of Figure~\ref{fig:ENDsenspar2} displays quartiles of the distribution of the confounding level when omitting $d=8$ covariates over time. We find that the estimated robustness value at $t = 1$ is at $40$th percentile of the \textit{leave-8-out} confounding. Given these benchmarking results, we conclude that there would need to be substantial unobserved confounding to change the sign of the estimated effect at $t = 1$, but only a small amount of unobserved confounding to make the estimated effect at $t=1$ statistically insignificant. We note that the empirical range for $|\rho_{n,-R}(t)|$ was $[0.31, 0.52]$ for $t \in [1/12,1]$ when $|R| = 8$. Narrower bounds for the causal effect could be obtained under the assumption that $|\rho_c(t)| \leq 0.52$. The minimum influential robustness value needs to be MIRV$_{n,.05}(1, 0) = 0.015$ to turn the effect insignificant at $t=1$. We find that \textit{surgery}, \textit{age} and \textit{high T stage} are the only three covariates whose estimated confounding level are larger than the MIRV at time $t = 1$ year. As the bounds become sharper, the estimated effect are more robust and less likely to be explained away by a small amount of unobserved confounding.

\section{Conclusion}
In this article, we developed a nonparametric sensitivity analysis framework to assess the robustness of causal evidence to unobserved confounding for time-to-event outcomes. We focused on the difference between counterfactual survival curves and restricted mean survival times, which has a clear causal interpretation, unlike hazard ratios, the most commonly used metric in time-to-event settings. We provided estimators of the effect bounds for which valid pointwise and uniform inference can be obtained without requiring correctly specified parametric or semiparametric models for the distribution of the observed data or the structure of unobserved confounding. This flexibility is useful because the nature of the data structure is often unknown in practice. Our proposed methods provide  a practical approach for scientific researchers to understand amount of unobserved confounding needed to change the causal conclusion. In addition, we provided tools for assessing the plausibility of the results based on observed confounding. One assumption we did rely on is that of no unobserved confounders between the outcome and censoring, which can be violated in practice \citep{huang2008regression}. An extension of our framework to assess robustness to dependent censoring may be of interest for future research.

\bibliographystyle{biom}
\bibliography{references}

\clearpage

\begin{adjustwidth}{-.25in}{-.25in}

\begin{appendix}

\begin{center}
\huge{Supplementary Material}
\end{center}

\small

\section{Proof of Theorems}\label{app:estinf}
\begin{proof}[\bfseries{Proof of Proposition~\ref{prop:par_id}}]
We can write
\begin{align*}
E\left[ g_{c,t} \alpha_c\right] &= E\left[ E \left\{ I(T>t)\mid A, W, U \right\} \frac{A}{\pi_c(W, U)}- E\left\{ I(T>t)\mid A, W, U \right\} \frac{1-A}{1-\pi_c(W, U)} \right] \\
&= E\left[ E \left\{ I(T>t)\mid A = 1, W, U \right\} \frac{A}{\pi_c(W, U)}- E\left\{ I(T>t)\mid A = 0, W, U \right\}\frac{1-A}{1-\pi_c(W, U)} \right]\\
& = E\left[ E \left\{ I(T>t)\mid A=1, W, U \right\} - E\left\{ I(T>t)\mid A=0, W, U \right\} \right] \\
&= \theta_c(t)
\end{align*}
and
\begin{align*}
E\left[ g_{P,t} \alpha_P\right] &= E\left[ E \left\{ I(T>t)\mid A, W\right\} \frac{A}{\pi_P(W)}- E\left\{ I(T>t)\mid A, W \right\}\frac{1-A}{1-\pi_P(W)} \right] \\
&= E\left[ E \left\{ I(T>t)\mid A = 1, W \right\} \frac{A}{\pi_P(W)}- E\left\{ I(T>t)\mid A = 0, W \right\}\frac{1-A}{1-\pi_P(W)} \right] \\
&= E\left[ E_P\left\{ I(T>t)\mid A=1, W \right\} - E_P\left\{ I(T>t)\mid A=0, W\right\} \right] \\
&= \theta_P(t).
\end{align*}
We can similarly find $E\left[ g_{P,t} \alpha_c\right]  = E\left[ g_{c,t} \alpha_P\right]= \theta_P(t)$. Therefore,
\[\theta_c(t) - \theta_P(t) =E\left[ \left\{g_{c,t}-g_{P,t}\right\}\left\{ \alpha_c - \alpha_P\right\} \right].\]
Furthermore,  $E \left[g_{c,t} \right] = E \left[ g_{P,t} \right] = P(T >t),$ so $\n{Var}(g_{c,t} - g_{P,t}) = E\left[\left\{g_{c,t}-g_{P,t}\right\}^2\right]$, and $E\left[ \alpha_c\right] = E\left[ \alpha_P \right] = 0$, so $\n{Var}(\alpha_{c} - \alpha_{P}) = E\left[\left\{\alpha_{c}-\alpha_{P}\right\}^2\right]$. We now have
\begin{align*}
    \left[\theta_P(t) - \theta_c(t)\right]^2 &= \left[ \n{Cor}\left(g_{c,t}-g_{P,t}, \: \alpha_c-\alpha_P\right)\right]^2 E\left[\left\{g_{c,t}-g_{P,t}\right\}^2\right]E\left[\left\{\alpha_c - \alpha_P\right\}^2\right] \\
    &= \psi_P(t) \tau_P \left[ \rho_c(t)\right]^2 \frac{E\left[\left\{g_{c,t}-g_{P,t}\right\}^2\right]}{\psi_P(t)} \frac{E\left[\left\{\alpha_c - \alpha_P\right\}^2\right]}{\tau_P}.
\end{align*}
We note that 
\[E\left[ \left\{ I(T > t) - g_{P,t}\right\}^2 \right] =  E[ g_{P, t} \{1 - g_{P,t}\}] = \psi_P(t),\]
so if $\psi_P(t) = 0$, then $ E[ I(T>t) \mid A , W] = I(T > t)$ with probability one, so $\theta_P(t) = \theta_c(t) = 0$. Similarly, if $\tau_P = 0$, then $\pi_P(W) = 0$ or 1 with probability 1, which violates the positivity assumption.

It remains to show that $E\left[\left\{g_{c,t}-g_{P,t}\right\}^2\right] / \psi_P(t) = s_{c,T}(t)$ and $E\left[\left\{\alpha_c - \alpha_P\right\}^2\right] / \tau_P = s_{c,A} / [1-s_{c,A}]$. Let $\omega_t := P(T > t)$. We have
\begin{align*}
     \frac{E\left[\left\{g_{c,t}-g_{P,t}\right\}^2\right]}{\psi_P(t)} &= \frac{E\left[\left\{g_{c,t} - \omega_t + \omega_t - g_{P,t}\right\}^2\right]}{E\left[\left\{g_{P,t}- \omega_t + \omega_t - I(T > t)\right\}^2\right]} \\
     &= \frac{E\left[\left\{ g_{c,t}  - \omega_t\right\}^2\right] + E\left[\left\{g_{P,t} - \omega_t\right\}^2\right] - 2E\left[\left\{g_{c,t} - \omega_t \right\} \left\{ g_{P,t} - \omega_t\right\}\right]  }{E\left[\left\{g_{P,t}- \omega_t\right\}^2\right] + E\left[\left\{I(T > t) - \omega_t\right\}^2\right] - 2 E\left[\left\{g_{P,t}- \omega_t\right\} \left\{ I(T > t) - \omega_t\right\}\right]}.
\end{align*}
Since $E[ g_{c,t}] = E[ g_{P, t}] = E[I(T>t)] = \omega_t$, $E\left[\left\{ g_{c,t}  - \omega_t\right\}^2\right] = \n{Var}(g_{c,t})$, $E\left[\left\{ g_{P,t}  - \omega_t\right\}^2\right] = \n{Var}(g_{P,t})$, and $E\left[\left\{I(T > t) - \omega_t\right\}^2\right] = \n{Var}(I(T>t))$. In addition, by the tower property, $E[ g_{c,t} \mid A, W] = g_{P,t}$, so that $E\left[\left\{g_{c,t} - \omega_t \right\} \left\{ g_{P,t} - \omega_t\right\}\right] = E\left[ \left\{ g_{P,t} - \omega_t\right\}^2\right] = \n{Var}(g_{P,t})$. Similarly, $E\left[\left\{g_{P,t}- \omega_t\right\} \left\{ I(T > t) - \omega_t\right\}\right] = E\left[\left\{g_{P,t}- \omega_t\right\}^2 \right] = \n{Var}(g_{P,t})$. Hence, 
\[ \frac{E\left[\left\{g_{c,t}-g_{P,t}\right\}^2\right]}{\psi_P(t)} = \frac{\n{Var}(g_{c,t}) - \n{Var}(g_{P,t})}{\n{Var}(I(T>t)) -  \n{Var}(g_{P,t})} = s_{c,T}(t) \]
as desired.

We note that $s_{c,A} / [1-s_{c,A}] = \left\{E[ \alpha_c^2] - E[ \alpha_P^2]\right\} / E[ \alpha_P^2]$. In addition, we have
\begin{align*}
    E \left[ \alpha_c \alpha_P \right] &= E\left[ \left\{ \frac{A}{\pi_c(W,U)}- \frac{1-A}{1-\pi_c(W,U)} \right\} \left\{ \frac{A}{\pi_P(W)}- \frac{1-A}{1-\pi_P(W)} \right\}\right] \\
    &= E\left[ \frac{A}{\pi_c(W,U)\pi_P(W)} + \frac{1-A}{\left\{1-\pi_c(W,U)\right\} \left\{1-\pi_P(W)\right\}}\right] \\
    &= E\left[ \frac{1}{\pi_P(W)} + \frac{1}{1-\pi_P(W)}\right] \\
    &= E\left[ \frac{1}{\pi_P(W) \left\{1-\pi_P(W)\right\} }\right] \\
    &= E\left[ \alpha_P^2 \right].
\end{align*}
Thus, we have
\begin{align*}
    \frac{E\left[\left\{\alpha_c - \alpha_P\right\}^2\right]}{\tau_P} &= \frac{E\left[\alpha_c^2\right] + E\left[\alpha_P^2\right] - 2 E\left[\alpha_c \alpha_P\right]}{E\left[ \alpha_P^2 \right]} \\
    &= \frac{E\left[\alpha_c^2\right] - E\left[\alpha_P^2\right]}{E\left[ \alpha_P^2 \right]} \\
    &= \frac{s_{c,A}}{1-s_{c,A}}.
\end{align*}
\end{proof}

\begin{proof}[\bfseries{Proof of Propositions~\ref{prop:eif1} and~\ref{prop:eif2}}] 
Let $\{ P_\epsilon : |\epsilon| \leq \delta\}$ be a differentiable in quadratic mean path with $P_{\epsilon = 0} = P$ and score function $\dot\ell_P \in L_2(P)$ at $\epsilon = 0$. For a distribution $P$ of $(Y, \Delta, A, W)$, we let $Q$ be the marginal distribution of $(A,W)$ as implied by $P$. We have
\begin{align*}
\psi_{P}(t) &= E[I(T>t)-S_P(t \mid A, W)]^2\\
&= E[I(T>t)-2 1(T>t)S_P(t \mid A, W)+S_P(t \mid A, W)^2] \\
&= E[S_P(t \mid A, W)]-2 E[E(I(T>t) \mid A, W) S_P(t \mid A, W)] + E \left[S_P(t \mid A, W)^2\right] \\
&= E[S_P(t \mid A, W)\{1-S_P(t \mid A, W)\}].
\end{align*}
By the product rule, we then have
\begin{align*}
    \left.\frac{\partial}{\partial \epsilon} \psi_\epsilon(t)\right|_{\epsilon=0} &=\left.\frac{\partial}{\partial \epsilon} \int S_\epsilon(t \mid a, w)\left\{1-S_\epsilon(t \mid a, w)\right\} d Q_\epsilon(a, w)\right|_{\epsilon=0} \\
&=\left.\int \frac{\partial}{\partial \epsilon} S_\epsilon(t \mid a, w)\left\{1-S_\epsilon(t \mid a, w)\right\}\right|_{\epsilon=0} \, d Q_P(a, w)\\
&\qquad +\int S_P(t \mid a, w)\left\{1-S_P(t \mid a, w)\right\} \dot\ell_P(a, w) \, d Q_P(a, w) \\
&=\left.\int\left\{1-2 S_P(t \mid a, w)\right\} \frac{\partial}{\partial \epsilon} S_\epsilon(t \mid a, w)\right|_{\epsilon=0} \, d Q_P(a, w)\\
&\qquad +\int S_P(t \mid a, w)\left\{1-S_P(t \mid a, w)\right\} \dot\ell_P(a, w) \, d Q_P(a, w).
\end{align*}
Based on the proof in \cite{westling2023inference}, the first term can be written as
\begin{align*}
    &  \left.\int\left\{1-  2 S_P(t \mid a, w)\right\} \frac{\partial}{\partial \epsilon} S_\epsilon(t \mid a, w)\right|_{\epsilon=0} d Q_P(a, w) \\
    & \quad = E_P\bigg[\{1-S_P(t \mid A, W)\} S_P(t \mid A, W) \bigg\{H_P(t \wedge Y, A, W) - \frac{I(Y \leq t, \Delta=1) S_P(Y- \mid A, W)}{S_P(Y \mid A, W) R_P(Y \mid A, W)}\bigg\} \dot{\ell}_P(Y, \Delta, A, W)\bigg]
\end{align*}
where $H_P(u, a, w):=\int_0^u S_P(u-\mid a, w) / [S_P(u \mid a, w) R_P(u \mid a, w)^2]  \,  F_P(d u \mid a, w)$.  Combining these results and simplifying using $F_P\left(d u \mid a, w\right) / R_P\left(u \mid a, w\right)=\Lambda_P\left(d u \mid a, w\right)$ and $R_P\left(u \mid a, w\right)=S_P\left(u-\mid a, w\right) G_P\left(u \mid a, w\right)$, we find that the uncentered EIF of $\psi_P(t)$ at $P$ is
\begin{align*}
     D_{P,\psi,t}(y, \delta, a, w) 
     &= S_P\left(t \mid a, w\right)\left[1-2S_P\left(t \mid a, w\right)\right]  \left\{ -\frac{I(y \leq t, \delta=1)}{S_P\left(y \mid a, w\right) G_P\left(y \mid a, w\right)} \right. \\
     &\qquad  + \left. \int_P^{t \wedge y} \frac{\Lambda_P\left(d u \mid a, w\right)}{S_P\left(u \mid a, w\right) G_P\left(u \mid a, w\right)} \right\} 
      + S_P\left(t \mid a, w\right) \left[1-S_P\left(t \mid a, w\right)\right].
\end{align*}
Similarly, we have
\begin{align*}
 E[\alpha_P^2]  &= E\left[\left\{\frac{A}{\pi_P(W)}-\frac{1-A}{1-\pi_P(W)}\right\}^2\right] \\
&= E\left[\left\{\frac{A-\pi_P(W)}{\pi_P(W)[1-\pi_P(W)]}\right\}^2\right] \\
&= E\left[E\left\{\left.\left[\frac{A-\pi_P(W)}{\pi_P(W)\{1-\pi_P(W)\}}\right]^2 \right\rvert\, W\right\}\right] \\
&= E\left[\frac{E\left\{[A-\pi_P(W)]^2 \mid W\right\}}{\pi_P(W)^2\{1-\pi_P(W)\}^2} \right] \\
&= E\left[\frac{1}{\pi_P(W)\{1-\pi_P(W)\}}\right].
\end{align*}
By the product rule, we have
\begin{align*}
\left.\frac{\partial}{\partial \epsilon} \tau_\epsilon \right|_{\epsilon=0} & =\left.\frac{\partial}{\partial \epsilon} \int \frac{1}{\pi_\epsilon(w)\left\{1-\pi_\epsilon(w)\right\}} \, d Q_\epsilon(w)\right|_{\epsilon=0} \\
& =\left.\int \frac{\partial}{\partial \epsilon} \frac{1}{\pi_\epsilon(w)\left\{1-\pi_\epsilon(w)\right\}}\right|_{\epsilon=0} \, d Q_P(w)+\int \frac{1}{\pi_P(w)\left\{1-\pi_P(w)\right\}} \dot\ell_P(w) \, d Q_P(w).
\end{align*}
We can write the first term as
\begin{align*}
& \left.\int \frac{2 \pi_P(w)-1}{\left[\pi_P(w)\{1-\pi_P(w)\}\right]^2} \frac{\partial}{\partial \epsilon} \pi_\epsilon(w)\right|_{\epsilon=0}\, d Q_P(w) \\
 & \quad \quad =\int \frac{2 \pi_P(w)-1}{\left[\pi_P(w)\{1-\pi_P(w)\}\right]^2} \int I(a=1) \dot\ell_P(a \mid w) \, P(d a \mid w) \, d Q_P(w) \\
 & \quad \quad =\iint I(a=1) \frac{2 \pi_P(w)-1}{\left[\pi_P(w)\{1-\pi_P(w)\}\right]^2} \dot\ell_P(a \mid w) \, d P(a, w) \\
 & \quad \quad =E_P\left[\frac{A \left\{2 \pi_P(W)-1\right\}}{\left\{\pi_P(w)[1-\pi_P(w)]\right\}^2} \dot\ell_P(A \mid W)\right].
\end{align*}
We note that
\begin{align*}
E_P\left[\left.\frac{A \left\{2 \pi_P(W)-1\right\}}{\left\{\pi_P(w)[1-\pi_P(w)]\right\}^2} \right\rvert\, W=w\right]=\frac{2 \pi_P(w)-1}{\pi_P(w)\left[1-\pi_P(w)\right]^2}.
\end{align*}
Therefore, by properties of score functions and the tower property, we have
\begin{align*}
    & \left.\int \frac{\partial}{\partial \epsilon} \frac{1}{\pi_\epsilon(w)\left\{1-\pi_\epsilon(w)\right\}}\right|_{\epsilon=0} \, d Q_P(w) \\
    &= E_P\left[\left\{\frac{A \left[2 \pi_P(W)-1\right]}{\left[\pi_P(W)\{1-\pi_P(W)\}\right]^2}-\frac{2 \pi_P(W)-1}{\pi_P(W)[1-\pi_P(W)]^2}\right\} \dot\ell_P(A, W)\right].
\end{align*}
Combing these results, we find that the uncentered EIF of $\tau_P$ at $P$ is
\begin{align*}
    D_{P, \tau}(a,w) &=  \frac{a  [2\pi_P(w)-1]}{\left[\pi_P(w)\{1-\pi_P(w)\}\right]^2}- \frac{2 \pi_P(w)-1}{\pi_P(w)[1-\pi_P(w)]^2} + \frac{1}{\pi_P(w)[1-\pi_P(w)]} \\
    &=  \frac{2}{\pi_P(w)[1-\pi_P(w)]} - \frac{[a -\pi_P(w)]^2}{\left[\pi_P(w)\{1-\pi_P(w)\}\right]^2}.
\end{align*}

\end{proof}

We define $D_{\infty, \psi, t}$ as the influence function of $\psi_P(t)$ with nuisance functions $S_P$ and $G_\infty$.
\begin{lemma}\label{lemma:psi_rep}
If~\ref{itm:positivity} holds, then $PD_{\infty, \psi, t}^\ast = 0$ and $D_{n,k, \psi, t}-D_{\infty, \psi, t} = \sum_{j=1}^6 U_{n,k,j,t}$ for functions $U_{n,k,j,t}$ defined in the proof that satisfy the following bounds:
\begin{align*}
    P U_{n, k, 1, t}^2 &\leq \eta^2 E_P\left[\sup _{u \in[0, t]} \left| \frac{S_{n, k}\left(t \mid A, W\right)}{S_{n, k}\left(u \mid A, W\right)}-\frac{S_{P}\left(t \mid A, W\right)}{S_{P}\left(u \mid A, W\right)}\right|^2\right] \\
    P U_{n, k, 2, t}^2 &\leq E_P\left[\sup _{u \in[0, t]}\left|\frac{1}{G_{n, k}\left(u \mid A, W\right)}-\frac{1}{G_{\infty}\left(u \mid A, W\right)}\right|^2\right] \\
    P U_{n, k, 3, t}^2 &\leq (4\eta^2 +1)E_P\left[\sup _{u \in[0, t]}\left|\frac{S_{n, k}\left(t \mid A,  W\right)}{S_{n, k}\left(u \mid A,  W\right)}-\frac{S_{P}\left(t \mid A,  W\right)}{S_{P}\left(u \mid A,  W\right)}\right|^2\right] \\
     P U_{n, k, 4, t}^2 &\leq E_P\left[\sup _{u \in[0, t]}\left|\frac{1}{G_{n, k}\left(u \mid A, W\right)}-\frac{1}{G_{\infty}\left(u \mid A, W\right)}\right|^2\right] \\
     P U_{n, k, 5, t}^2 & \leq E_P\left[\sup _{u \in[0, t]}\left|\frac{1}{G_{n, k}\left(u \mid A, W\right)}-\frac{1}{G_{\infty}\left(u \mid A, W\right)}\right|^2\right] \\
     P U_{n, k, 6, t}^2 & \leq 4\eta^2 E_P\left[\sup _{u \in[0, t]}\left|\frac{S_{n, k}\left(t \mid A,  W\right)}{S_{n, k}\left(u \mid A,  W\right)}-\frac{S_{P}\left(t \mid A,  W\right)}{S_{P}\left(u \mid A,  W\right)}\right|^2\right].
\end{align*}
\end{lemma}
\begin{proof}[\bfseries{Proof of Lemma~\ref{lemma:psi_rep}}]
We can write 
\[ D_{P,\psi,t}(y, \delta, a, w) =- S_P(t \mid a, w)  \left[ 1 - 2 S_P(t \mid a, w)\right] H_{S_P, G_P, t}(y, \delta, a, w) +  S_P(t \mid a, w)\left[ 1 - S_P(t \mid a, w) \right]\]
where 
\[H_{S,G,t}(y, \delta, a, w) :=\frac{I(y \leq t, \delta=1)}{S\left(y \mid a, w\right) G\left(y \mid a, w\right)}-\int_0^{t \wedge y} \frac{\Lambda\left(du \mid a, w\right)}{S_P\left(u \mid a, w\right) G_P\left(u \mid a, w\right)}.\]
As in Lemma~1 of \cite{westling2023inference}, we can show that
\[ E_P \left[H_{S, G, t}(Y, \Delta, A, W) \mid A = a, W =w \right] = -\int_0^t \frac{S_P(y- \mid a, w) G_P(y \mid a, w)}{S(y \mid a, w) G(y \mid a, w)} \, \left( \Lambda - \Lambda_P\right)(dy \mid a, w)\]
for any conditional survival function $S$ and corresponding cumulative hazard $\Lambda$ and conditional survival $G$. We note that there is a typo in the proof of Lemma~1 of \cite{westling2023inference}: $E_0\left[ H_{S, G, t, a_0}(Y, \Delta, W) \mid W =w \right]$ should be $E_0 \left[ H_{S, G, t, a_0}(Y, \Delta, W) \mid A = a, W =w\right]$. Dropping the conditioning  on $A, W$ for notational simplicity, we now have
\begin{align*}
    P D_{S,G_\infty, \psi, t} &= E_P\left\{  - S(t \mid A ,W)\left[ 1 - 2 S(t \mid A, W)\right] H_{S, G_\infty, t}(Y, \Delta, A, W) \right\}\\
    &\qquad + E_P \left\{ S(t \mid A ,W)  \left[ 1 - S(t \mid A, W) \right] \right\}  \\
    &= E_P\left\{  S(t) \left[ 1 - 2 S(t)\right] \int_0^t \frac{S_P(y-) G_P(y)}{S(y ) G_\infty(y)} \, \left( \Lambda - \Lambda_P\right)(dy)\right\}  + E_P\left\{ S(t) \left[ 1 - S(t) \right] \right\} \\
    &= E_P\left\{  S(t) \left[ 1 - 2 S(t)\right] \int_0^t \frac{S_P(y-)}{S(y ) }  \left[ \frac{ G_P(y)}{G_\infty(y)} - 1\right]\, \left( \Lambda - \Lambda_P\right)(dy)\right\} \\
    &\qquad +  E_P\left\{  S(t) \left[ 1 - 2 S(t)\right] \int_0^t \frac{S_P(y-)}{S(y) } \, \left( \Lambda - \Lambda_P\right)(dy)\right\}  + E_P\left\{ S(t) \left[ 1 - S(t) \right] \right\}.
\end{align*}
By the Duhamel equation (Theorem~6 of \citeauthor{gill1990survey}, \citeyear{gill1990survey}), we have
\[ S(t \mid A ,W)  \int_0^t \frac{S_P(y- \mid A, W)}{S(y \mid A, W) } \, \left( \Lambda - \Lambda_P\right)(dy \mid A, W) = - \left[  S(t \mid A ,W) -  S_P(t \mid A ,W) \right].\]
Hence, 
\begin{align*}
    P D_{S,G_\infty, \psi, t}&= E_P\left\{  S(t) \left[ 1 - 2 S(t)\right] \int_0^t \frac{S_P(y-)}{S(y ) }  \left[ \frac{ G_P(y)}{G_\infty(y)} - 1\right]\, \left( \Lambda - \Lambda_P\right)(dy )\right\} \\
    &\qquad +  E_P\left\{  -\left[ 1 - 2 S(t)\right] \left[ S(t) - S_P(t) \right]\right\} + E_P\left\{ S(t) \left[ 1 - S(t) \right] \right\} \\
    &= E_P\left\{  S(t) \left[ 1 - 2 S(t)\right] \int_0^t \frac{S_P(y-)}{S(y) }  \left[ \frac{ G_P(y)}{G_\infty(y)} - 1\right]\, \left( \Lambda - \Lambda_P\right)(dy )\right\} \\
    &\qquad +  E_P\left\{ \left[ S(t) - S_P(t)\right]^2 \right\} + \psi_P(t).
\end{align*}
When $S = S_P$, the first and second term are zero. Therefore, $P D_{S_P,G_\infty, \psi, t} = \psi_{P}(t)$, so that $PD_{S_P,G_\infty, \psi, t}^\ast = 0$.

By adding and subtracting terms, we can write $D_{n,k, \psi, t}-D_{\infty, \psi, t} = \sum_{j=1}^6 U_{n,k,j,t}$ where
\begin{align*}
    U_{n,k,1,t} &:= \left[ \frac{ S_{n,k}(t \mid a, w)}{S_{n,k}(y \mid a, w)} - \frac{ S_{P}(t \mid a, w)}{S_{P}(y \mid a, w)} \right]\frac{2 S_{n,k}(t \mid a, w)-1}{G_{n,k}(y \mid ,a ,w)} I(y \leq t, \delta=1)  \\
    U_{n,k,2,t} &:=  \left[ \frac{1}{G_{n,k}(y \mid ,a ,w)}  - \frac{1}{G_{\infty}(y \mid ,a ,w)} \right] \left[ 2 S_{n,k}(t \mid a, w)-1\right] \frac{S_P(t \mid a, w)}{S_{P}(y \mid a, w)} I(y \leq t, \delta=1) \\
    U_{n,k,3,t} &:=\left[S_{n,k}(t \mid a, w) -  S_{P}(t \mid a, w)\right] \left[ 2 \frac{ S_{P}(t \mid a, w)}{S_{P}(y \mid a, w)} \frac{I(y \leq t, \delta=1)}{G_{\infty}(y \mid ,a ,w)} + 1 -  S_{n,k}(t \mid a, w) -  S_{P}(t \mid a, w) \right] \\
    U_{n,k,4,t} &:=\int_0^{t \wedge y}  \left[ \frac{ S_{n,k}(t \mid a, w)}{S_{n,k}(u \mid a, w)} \Lambda_{n,k}(du \mid a, w)- \frac{ S_{P}(t \mid a, w)}{S_{P}(u \mid a, w)} \Lambda_{P}(du \mid a, w)\right]\frac{1 -  2 S_{n,k}(t \mid a, w)}{G_{n,k}(u \mid ,a ,w)}  \\
    U_{n,k,5,t} &:=\int_0^{t \wedge y} \left[ \frac{1}{G_{n,k}(u \mid ,a ,w)}  - \frac{1}{G_{\infty}(u \mid ,a ,w)} \right] \left[ 1 - 2 S_{n,k}(t \mid a, w)\right]\frac{S_P(t \mid a, w)}{S_{P}(u \mid a, w)} \Lambda_P(du \mid a, w) \\
    U_{n,k,6,t} &:=-2\left[S_{n,k}(t \mid a, w) -  S_{P}(t \mid a, w)\right]\int_0^{t \wedge y} \frac{S_{P}(t \mid a, w)}{S_{P}(u \mid a, w)} \frac{\Lambda_P(du \mid a, w)}{G_\infty(u \mid a, w)}.
\end{align*}
The bounds for each of these terms follow by the derivations in Lemma~3 of \cite{westling2023inference} and condition~\ref{itm:positivity}.
\end{proof}

\begin{lemma}\label{lemma:tau_rep}
If~\ref{itm:positivity} holds, then $P\left(D_{n,k, \tau}-D_{P, \tau}\right)^2 \leq 25\eta^{12}  P\left( \pi_{n,k} - \pi_P \right)^2$, and 
\begin{align*}
    P\left(D_{n,k, \tau}-D_{P, \tau}\right) &= -E_P\left\{\left[\pi_{n,k}(W) - \pi_P(W)\right]^2  \frac{\left[\pi_{n,k}(W) +\pi_{P}(W)-1 \right]^2 +\pi_{P}(W)\left[1 - \pi_{P}(W)\right]}{\pi_{n,k}(W)^2\left[1 - \pi_{n,k}(W)\right]^2\pi_{P}(W)\left[1 - \pi_{P}(W)\right]} \right\}.
\end{align*}
\end{lemma}
\begin{proof}[\bfseries{Proof of Lemma~\ref{lemma:tau_rep}}]
We can write
\begin{align*}
D_{n,k, \tau}-D_{P, \tau} &= \frac{2}{\pi_{n,k}\left[1 - \pi_{n,k}\right]}  - \frac{2}{\pi_{P}\left[1 - \pi_{P}\right]}  - \frac{\left[A - \pi_{n,k}\right]^2 }{\pi_{n,k}^2\left[1 - \pi_{n,k}\right]^2} +\frac{\left[A - \pi_{P}\right]^2}{\pi_{P}^2\left[1 - \pi_{P}\right]^2}  \\
&= 2\frac{\pi_{P}\left[1 - \pi_{P}\right] - \pi_{n,k}\left[1 - \pi_{n,k}\right]}{\pi_{n,k}\left[1 - \pi_{n,k}\right]\pi_{P}\left[1 - \pi_{P}\right]}   - \frac{\left[A - \pi_{n,k}\right]^2 - \left[A - \pi_{P}\right]^2}{\pi_{n,k}^2\left[1 - \pi_{n,k}\right]^2} \\
&\qquad - \left[A - \pi_{P}\right]^2\left\{\frac{1}{\pi_{n,k}^2\left[1 - \pi_{n,k}\right]^2} - \frac{1}{\pi_{P}^2\left[1 - \pi_{P}\right]^2} \right\}  \\
&= 2\frac{\left[ \pi_{n,k} - \pi_{P}\right]\left[\pi_{P} + \pi_{n,k} - 1\right]}{\pi_{n,k}\left[1 - \pi_{n,k}\right]\pi_{P}\left[1 - \pi_{P}\right]} + \frac{\left[\pi_{n,k} - \pi_P\right]\left[2A - \pi_{n,k} -\pi_{P}\right]}{\pi_{n,k}^2\left[1 - \pi_{n,k}\right]^2}  \\
&\qquad - \left[A - \pi_{P}\right]^2\frac{\left[\pi_{n,k} - \pi_P\right] \left[\pi_{n,k} +\pi_{P} - 1 \right] \left[\pi_{P}\left\{1 - \pi_{P}\right\} + \pi_{n,k}\left\{1 - \pi_{n,k}\right\}\right]}{\pi_{n,k}^2\left[1 - \pi_{n,k}\right]^2\pi_{P}^2\left[1 - \pi_{P}\right]^2}  \\
&= \left\{\pi_{n,k} - \pi_P\right\} \left\{ 2\frac{\pi_{P} + \pi_{n,k} - 1}{\pi_{n,k}\left[1 - \pi_{n,k}\right]\pi_{P}\left[1 - \pi_{P}\right]} + \frac{2A - \pi_{n,k} -\pi_{P}}{\pi_{n,k}^2\left[1 - \pi_{n,k}\right]^2} \right. \\
&\qquad\qquad\qquad\qquad\left. - \left[A - \pi_{P}\right]^2\frac{\left[\pi_{n,k} +\pi_{P}-1 \right] \left[\pi_{P}\left\{1 - \pi_{P}\right\} + \pi_{n,k}\left\{1 - \pi_{n,k}\right\}\right]}{\pi_{n,k}^2\left[1 - \pi_{n,k}\right]^2\pi_{P}^2\left[1 - \pi_{P}\right]^2} \right\}.
\end{align*}
Thus 
\[ P\left(D_{n,k, \tau}-D_{P, \tau}\right)^2 \leq 25\eta^{12} P\left( \pi_{n,k} - \pi_P \right)^2 \] by~\ref{itm:positivity}. We furthermore have
\begin{align*}
    P\left(D_{n,k, \tau}-D_{P, \tau}\right) &= E_P\left[\left\{\pi_{n,k} - \pi_P\right\} \left\{ 2\frac{\pi_{P} + \pi_{n,k} - 1}{\pi_{n,k}\left[1 - \pi_{n,k}\right]\pi_{P}\left[1 - \pi_{P}\right]} + \frac{2A - \pi_{n,k} -\pi_{P}}{\pi_{n,k}^2\left[1 - \pi_{n,k}\right]^2} \right.\right. \\
    &\qquad\qquad\qquad\qquad\left.\left. - \left[A - \pi_{P}\right]^2\frac{\left[\pi_{n,k} +\pi_{P}-1 \right] \left[\pi_{P}\left\{1 - \pi_{P}\right\} + \pi_{n,k}\left\{1 - \pi_{n,k}\right\}\right]}{\pi_{n,k}^2\left[1 - \pi_{n,k}\right]^2\pi_{P}^2\left[1 - \pi_{P}\right]^2} \right\}\right] \\
    &= E_P\left[\left\{\pi_{n,k} - \pi_P\right\} \left\{ 2\frac{\pi_{P} + \pi_{n,k} - 1}{\pi_{n,k}\left[1 - \pi_{n,k}\right]\pi_{P}\left[1 - \pi_{P}\right]} + \frac{\pi_P - \pi_{n,k}}{\pi_{n,k}^2\left[1 - \pi_{n,k}\right]^2} \right.\right. \\
    &\qquad\qquad\qquad\qquad\left.\left. - \pi_P\left[1 - \pi_{P}\right]\frac{\left[\pi_{n,k} +\pi_{P}-1 \right] \left[\pi_{P}\left\{1 - \pi_{P}\right\} + \pi_{n,k}\left\{1 - \pi_{n,k}\right\}\right]}{\pi_{n,k}^2\left[1 - \pi_{n,k}\right]^2\pi_{P}^2\left[1 - \pi_{P}\right]^2} \right\}\right] \\
    &= E_P\left[\left\{\pi_{n,k} - \pi_P\right\} \left\{ \left[\pi_{n,k} +\pi_{P}-1 \right] \frac{\pi_{n,k}\left[1 - \pi_{n,k}\right] - \pi_{P}\left[1 - \pi_{P}\right]}{\pi_{n,k}^2\left[1 - \pi_{n,k}\right]^2\pi_{P}\left[1 - \pi_{P}\right]} + \frac{\pi_P - \pi_{n,k}}{\pi_{n,k}^2\left[1 - \pi_{n,k}\right]^2} \right\} \right] \\
    &= -E_P\left\{\left[\pi_{n,k} - \pi_P\right]^2  \frac{\left[\pi_{n,k} +\pi_{P}-1 \right]^2 +\pi_{P}\left[1 - \pi_{P}\right]}{\pi_{n,k}^2\left[1 - \pi_{n,k}\right]^2\pi_{P}\left[1 - \pi_{P}\right]} \right\}.
\end{align*}

\end{proof}

\begin{proof}[\bfseries{Proof of Theorem~\ref{thm:consistency}}]
Conditions~\ref{itm:nuisance}--\ref{itm:positivity} imply conditions (B1)--(B3) of \cite{westling2023inference} with $S_\infty = S_P$ and $\pi_\infty = \pi_P$. Thus, pointwise and uniform consistency of $\theta_n(t)$ follow by  Theorem~2 of \cite{westling2023inference}.

We now turn to consistency of $\psi_n$. If $S_\infty = S_P$, then adding and subtracting terms, we have
\begin{equation}
\begin{aligned}\label{eq:psi_decomp}
    \psi_n(t)-\psi_P(t) &= \mathbb{P}_n D_{\infty, \psi, t}^\ast + \frac{1}{K} \sum_{k=1}^K \frac{K n_k^{1 / 2}}{n} \mathbb{G}_n^k\left(D_{n,k, \psi, t}-D_{\infty, \psi, t}\right) +\frac{1}{K} \sum_{k=1}^K \frac{K n_k}{n}P \left(D_{n,k, \psi, t}-D_{\infty, \psi, t}\right).
\end{aligned}
\end{equation}
Since $P D_{\infty, \psi, t}^\ast = 0$ by Lemma~\ref{lemma:psi_rep}, the first term is $\fasterthan(1)$ by the weak law of large numbers. In addition, by standard derivations for cross-fitted empirical process terms (e.g., the derivations in Lemma~6 of  \citealp{westling2023inference}), we can show that
\begin{align*}
    E_P\left|  \frac{1}{K} \sum_{k=1}^K \frac{K n_k^{1 / 2}}{n} \mathbb{G}_n^k\left(D_{n,k, \psi, t}-D_{\infty,\psi, t}\right)\right| &\leq C n^{-1/2} \max_k E_P\left[ \left(D_{n,k, \psi, t}-D_{\infty, \psi, t}\right)^2\right]
\end{align*}
Since $D_{n,k, \psi, t}$ and $D_{\infty, \psi, t}$ are uniformly bounded, if $\max_k P\left(D_{n,k, \psi, t}-D_{\infty, \psi, t}\right)^2 = \fasterthan(1)$, then it follows that $\psi_n(t)\inprob \psi_P(t)$.   By Lemma~\ref{lemma:psi_rep}, $D_{n,k, \psi, t}-D_{\infty, \psi, t} = \sum_{j=1}^6 U_{n,k,j,t}$ and each of the bounds for these expressions provided in Lemma~\ref{lemma:psi_rep} is $\fasterthan(1)$ by~\ref{itm:nuisance},  so by the triangle inequality, $\max_k P\left(D_{n,k, \psi, t}-D_{\infty, \psi, t}\right)^2 = \fasterthan(1)$, which implies that $\psi_n(t) \inprob \psi_P(t)$. For uniform consistency, we have
\begin{align*}
\sup _{u \in[0, t]}\left|\psi_n(u)-\psi_P\left(u\right)\right| \leq &\sup _{u \in[0, t]}\left|\mathbb{P}_n D_{\infty, \psi,u}^*\right|+\sup _{u \in[0, t]}\left|\frac{1}{K} \sum_{k=1}^K \frac{K n_k^{1 / 2}}{n} \mathbb{G}_n^k\left(D_{n, k, \psi, u}-D_{\infty, \psi, u}\right)\right| \\
&+\sup _{u \in[0, t]}\left|\frac{1}{K} \sum_{k=1}^K \frac{K n_k}{n} P\left(D_{n, k, \psi, u}-D_{\infty, \psi, u}\right)\right| .
\end{align*}
Using Lemma~4 of \cite{westling2023inference}, Lemma~\ref{lemma:psi_rep}, the derivations above, and~\ref{itm:uniform}, we can show that each of these terms is $\fasterthan(1)$.

The proof for consistency of $\tau_n$ is similar.

We have
\begin{align}
    \tau_n-\tau_P &= \mathbb{P}_n D_{P, \tau}^\ast +\frac{1}{K} \sum_{k=1}^K \frac{K n_k^{1 / 2}}{n} \mathbb{G}_n^k\left(D_{n,k, \tau}-D_{P, \tau}\right) + \frac{1}{K} \sum_{k=1}^K \frac{K n_k}{n}P \left(D_{n,k, \tau}-D_{P, \tau}\right).\label{eq:tau_decomp}
\end{align}
Since $PD_{P,\tau}^\ast = 0$, $\mathbb{P}_n D_{P, \tau}^\ast = \fasterthan(1)$. By Lemma~\ref{lemma:tau_rep}, $P\left(D_{n,k, \tau}-D_{P, \tau}\right)^2 \leq C \eta^k P(\pi_{n,k} - \pi_P)^2$, which implies  by~\ref{itm:nuisance} that $\max_k  P\left(D_{n,k, \tau}-D_{P, \tau}\right)^2 \inprob 0$. Hence, the second term in~\eqref{eq:tau_decomp} is $\fasterthan(1)$. For the third term in~\eqref{eq:tau_decomp}, by Lemma~\ref{lemma:tau_rep} and~\ref{itm:positivity}, $\left| P\left(D_{n,k, \tau}-D_{P, \tau}\right) \right| \leq 2\eta^6 P \left( \pi_{n,k}-\pi_P \right)^2$, so that the third term in~\eqref{eq:tau_decomp} is $\fasterthan(1)$ by~\ref{itm:nuisance}. Thus, $\tau_n \inprob \tau_P$.
\end{proof}

\begin{proof}[\bfseries{Proof of Theorem~\ref{thm:asy_linear}}]
Conditions~\ref{itm:nuisance}--\ref{itm:asy} imply conditions (B1)--(B5) of \cite{westling2023inference}, and~\ref{itm:asy_unif} implies condition (B6) of \cite{westling2023inference}. Thus, pointwise and uniform asymptotic linearity of $\theta_n(t)$ follow by Theorem~2 of \cite{westling2023inference}.

Since $G_\infty = G_P$, as in the proof of Theorem~\ref{thm:consistency}, we have 
\begin{align}
    \psi_n(t)-\psi_P(t) - \mathbb{P}_n D_{P, \psi, t}^\ast = \frac{1}{K} \sum_{k=1}^K \frac{K n_k^{1 / 2}}{n} \mathbb{G}_n^k\left(D_{n,k, \psi, t}-D_{P, \psi, t}\right)+\frac{1}{K} \sum_{k=1}^K \frac{K n_k}{n}P\left(D_{n,k, \psi, t}-D_{P, \psi, t}\right). \label{eq:psi_rep2}
\end{align}
By Lemma~\ref{lemma:psi_rep}, $D_{n,k, \psi, t}-D_{P, \psi, t} = \sum_{j=1}^6 U_{n,k,j,t}$ and each of the bounds for these expressions provided in Lemma~\ref{lemma:psi_rep} is $\fasterthan(n^{-1/2})$ by~\ref{itm:asy},  so by the triangle inequality, $\max_k P\left(D_{n,k, \psi, t}-D_{P, \psi, t}\right)^2 = \fasterthan(n^{-1/2})$, 
As in the proof of Theorem~\ref{thm:consistency}, both terms in~\eqref{eq:psi_rep2} are then $\fasterthan(n^{-1/2})$. For uniform asymptotic linearity of $\psi_n$, by Lemma~\ref{lemma:psi_rep} and~\ref{itm:asy}--~\ref{itm:asy_unif} that
\begin{align*}
    \sup _{u \in[0, t]}\left|\frac{1}{K} \sum_{k=1}^K \frac{K n_k}{n}P\left(D_{n,k, \psi, t}-D_{P, \psi, t}\right)\right| = \fasterthan(n^{-1/2}).
\end{align*}
We can show that 
\begin{align*}
    \sup _{u \in[0, t]}\left|\frac{1}{K} \sum_{k=1}^K \frac{K n_k^{1 / 2}}{n} \mathbb{G}_n^k\left(D_{n,k, \psi, t}-D_{P, \psi, t}\right)\right| = \fasterthan(n^{-1/2})
\end{align*}
using Lemma~5 of \cite{westling2023inference} using the same basic argument as in Lemma~6 of \cite{westling2023inference}.

As in the proof of Theorem~\ref{thm:consistency}, we have
\begin{align}
    \tau_n-\tau_P -  \mathbb{P}_n D_{P, \tau}^\ast &= \frac{1}{K} \sum_{k=1}^K \frac{K n_k^{1 / 2}}{n} \mathbb{G}_n^k\left(D_{n,k, \tau}-D_{P, \tau}\right) + \frac{1}{K} \sum_{k=1}^K \frac{K n_k}{n}P \left(D_{n,k, \tau}-D_{P, \tau}\right).\label{eq:tau_decomp2}
\end{align}
By Lemma~\ref{lemma:tau_rep}, $P\left(D_{n,k, \tau}-D_{P, \tau}\right)^2 \leq 25 \eta^{12} P(\pi_{n,k} - \pi_P)^2$, which implies  by~\ref{itm:nuisance} that $\max_k  P\left(D_{n,k, \tau}-D_{P, \tau}\right)^2 \inprob 0$. Hence, since $\max_k K n_k^{1/2} /n = \boundeddet(n^{-1/2})$, the first term in~\eqref{eq:tau_decomp2} is $\fasterthan(n^{-1/2})$. For the second term in~\eqref{eq:tau_decomp2}, by Lemma~\ref{lemma:tau_rep} and~\ref{itm:positivity}, $\left| P\left(D_{n,k, \tau}-D_{P, \tau}\right) \right| \leq 2\eta^6 P \left( \pi_{n,k}-\pi_P \right)^2$, so that the second term in~\eqref{eq:tau_decomp2} is $\fasterthan(n^{-1/2})$ by~\ref{itm:asy}. Thus, $\tau_n-\tau_P =  \mathbb{P}_n D_{P, \tau}^\ast + \fasterthan(n^{-1/2})$.

We now have
\begin{align*}
\theta_{n,l}(t,v) - \theta_{P,l}(t,v)- \d{P}_n D_{P,l,t,v} &= \left[ \theta_n(t) - \theta_P(t) - \d{P}_n D_{P,\theta,t} \right] \\
&\qquad - v \left[ \sqrt{\psi_n(t) \tau_n} - \sqrt{\psi_P(t) \tau_P} -\frac{\tau_P \d{P}_nD_{P,\psi,t}^\ast + \psi_P(t) \d{P}_nD_{P,\tau}^\ast}{2\sqrt{\psi_P(t)\tau_P}} \right].
\end{align*}
We addressed the first term above. The second term can be written as
\begin{align*}
    &\sqrt{\psi_n(t) \tau_n} - \sqrt{\psi_P(t) \tau_P} -\frac{\tau_P \d{P}_nD_{P,\psi,t}^\ast + \psi_P(t) \d{P}_n D_{P,\tau}^\ast}{2\sqrt{\psi_P(t)\tau_P}}\\
    &\qquad = \left[\sqrt{\psi_n(t)} - \sqrt{\psi_P(t)} - \frac{ \d{P}_nD_{P,\psi,t}^\ast}{2\sqrt{\psi_P(t)}} \frac{\sqrt{\tau_P}}{\sqrt{\tau_n}}\right] \sqrt{\tau_n} + \left[\sqrt{\tau_n} - \sqrt{\tau_P} - \frac{ \d{P}_nD_{P,\tau}^\ast}{2\sqrt{\tau_P}}\right]\sqrt{\psi_P(t)} \\
    &\qquad = \left[\frac{\psi_n(t) - \psi_P(t)}{\sqrt{\psi_n(t)} + \sqrt{\psi_P(t)}} - \frac{ \d{P}_nD_{P,\psi,t}^\ast}{2\sqrt{\psi_P(t)}} \frac{\sqrt{\tau_P}}{\sqrt{\tau_n}}\right] \sqrt{\tau_n} + \left[\frac{\tau_n - \tau_P}{\sqrt{\tau_n} + \sqrt{\tau_P}} - \frac{ \d{P}_n D_{P,\tau}^\ast}{2\sqrt{\tau_P}}\right]\sqrt{\psi_P(t)} \\
    &\qquad = \left[\frac{\psi_n(t) - \psi_P(t) - \d{P}_nD_{P,\psi,t}^\ast}{\sqrt{\psi_n(t)} + \sqrt{\psi_P(t)}} + \left\{\frac{1}{\sqrt{\psi_n(t)} + \sqrt{\psi_P(t)}}- \frac{\sqrt{\tau_P} }{2\sqrt{\psi_P(t)\tau_n}} \right\}\d{P}_nD_{P,\psi,t}^\ast\right] \sqrt{\tau_n}\\
    &\qquad\qquad + \left[\frac{\tau_n - \tau_P -  \d{P}_n D_{P,\tau}^\ast}{\sqrt{\tau_n} + \sqrt{\tau_P}} + \left\{ \frac{1}{\sqrt{\tau_n} + \sqrt{\tau_P}} - \frac{1}{2\sqrt{\tau_P}}\right\}\d{P}_n D_{P,\tau}^\ast\right]\sqrt{\psi_P(t)} \\
     &\qquad = \left[\frac{\psi_n(t) - \psi_P(t) - \d{P}_nD_{P,\psi,t}^\ast}{\sqrt{\psi_n(t)} + \sqrt{\psi_P(t)}} + \frac{ \sqrt{\psi_P(t)} - \sqrt{\psi_n(t)}}{2\sqrt{\psi_n(t)\psi_P(t)} + 2 \psi_P(t)} \d{P}_nD_{P,\psi,t}^\ast  +\frac{\sqrt{\tau_n}- \sqrt{\tau_P} }{2\sqrt{\psi_P(t)\tau_n}}\d{P}_nD_{P,\psi,t}^\ast\right] \sqrt{\tau_n}\\
    &\qquad\qquad + \left[\frac{\tau_n - \tau_P -  \d{P}_n D_{P,\tau}^\ast}{\sqrt{\tau_n} + \sqrt{\tau_P}} + \frac{\sqrt{\tau_P} - \sqrt{\tau_n} }{2\sqrt{\tau_P\tau_n} + 2\tau_P}\d{P}_n D_{P,\tau}^\ast\right]\sqrt{\psi_P(t)}.
\end{align*}
If $\tau_P > 0$ and $\psi_P(t) > 0$, then since $\psi_n(t) - \psi_P(t) - \d{P}_nD_{P,\psi,t}^\ast = \fasterthan(n^{-1/2})$, $\tau_n - \tau_P - \d{P}_nD_{P,\tau}^\ast = \fasterthan(n^{-1/2})$, $\d{P}_nD_{P,\psi,t}^\ast = \bounded(n^{-1/2})$, $ \d{P}_nD_{P,\tau}^\ast  = \bounded(n^{-1/2})$, $\sqrt{\psi_n(t)} - \sqrt{\psi_P(t)} = \fasterthan(1)$, and $\sqrt{\tau_n} - \sqrt{\tau_P} = \fasterthan(1)$, the final expression above is $\fasterthan(n^{-1/2})$. If $\tau_P = 0$, then $D_{P,\tau}^\ast = 0$ as well, so that $\tau_n = \fasterthan(n^{-1/2})$ and $\sqrt{\tau_n} = \fasterthan(n^{-1/4})$. Similarly, if $\psi_P(t) = 0$, then $D_{P, \psi,t} = 0$ as well, so that $\psi_n(t) = \fasterthan(n^{-1/2})$ and $\sqrt{\psi_n(t)} =\fasterthan(n^{-1/4})$. Thus, recalling that we are interpreting $0/0 = 0$ for notational convenience, if both $\tau_P = 0$ and $ \psi_P(t) = 0$, then the above equals $\sqrt{\psi_n(t)\tau_n} = \fasterthan(n^{-1/2})$. If $\tau_P = 0$ and $\psi_P(t) > 0$, then the above equals
\begin{align*}
    \left[\frac{\psi_n(t) - \psi_P(t) - \d{P}_nD_{P,\psi,t}^\ast}{\sqrt{\psi_n(t)} + \sqrt{\psi_P(t)}} + \frac{ \sqrt{\psi_P(t)} - \sqrt{\psi_n(t)}}{2\sqrt{\psi_n(t)\psi_P(t)} + 2 \psi_P(t)} \d{P}_nD_{P,\psi,t}^\ast  +\frac{1}{2\sqrt{\psi_P(t)}}\d{P}_nD_{P,\psi,t}^\ast + \sqrt{\psi_P(t)}\right] \sqrt{\tau_n}.
\end{align*}
Pointwise asymptotic linearity of $\theta_{n,l}(t,v)$ follows, and a nearly identical calculation holds for $\theta_{n,u}(t,v)$.

For uniform asymptotic linearity, we have assumed that $\inf_{s \in[0,t]} \psi_P(s) > 0$, and as above if $\tau_P = 0$ then 
\begin{align*}
    &\sup_{s \in [0,t]} \left| \sqrt{\psi_n(s) \tau_n} - \sqrt{\psi_P(s) \tau_P} -\frac{\tau_P \d{P}_nD_{P,\psi,s}^\ast + \psi_P(s) \d{P}_n D_{P,\tau}^\ast}{2\sqrt{\psi_P(s)\tau_P}}\right| = \sup_{s \in [0,t]} \left| \sqrt{\psi_n(s) \tau_n}\right|,
\end{align*}
which is $\fasterthan(n^{-1/2})$ because $\sqrt{\tau_n}
= \boundeddet(n^{-1/2})$.
\end{proof}

\section{Additional inference results}\label{app:transformation}
\begin{proof}[\bfseries{Transformed pointwise confidence intervals}]
The function $g(x):=\text{log}[(1+x)/(1-x)]$ is chosen to map the range of the parameter from $(-1, 1)$ to $(-\infty, \infty)$. By the multivariate delta method, $n^{1/2} [g(\theta_{n,l}(t,v)) - g(\theta_{P,l}(t,v))]$ and $n^{1/2} [g(\theta_{n,u}(t,v)) - g(\theta_{P,u}(t,v))]$ converge jointly in distribution to a mean-zero bivariate normal distribution with variances $\tilde \sigma_{P,l,t,v}^2 := P ([2/(1-\theta_{P,l}^2(t,v))]D^\ast_{P,l, t,v})^2$ and $\tilde \sigma_{P,u,t,v}^2 := P ([2/(1-\theta_{P,u}^2(t,v))]D^\ast_{P,u, t,v})^2$, respectively, and covariance $\tilde\Sigma_{P, ul, t, v} := P ([4/(1-\theta_{P,l}^2(t,v))(1-\theta_{P,u}^2(t,v))]D^\ast_{P,u, t,v}D^\ast_{P,l, t,v})$. We then define the cross-fitted variance estimators $\tilde \sigma_{n,l,t,v} := [2/(1-\theta_{n,l}^2(t,v))] \sigma_{n,l,t,v} $, $\tilde \sigma_{n,u,t,v} := [2/(1-\theta_{n,u}^2(t,v))] \sigma_{n,u,t,v} $ and $\tilde \Sigma_{n, ul, t, v} := \left[2/\sqrt{(1-\theta_{n,l}^2(t,v))(1-\theta_{n,u}^2(t,v))}\right] \Sigma_{n, ul, t, v}$.

We then define an asymptotic $(1-\alpha)$-level Wald-type confidence interval as
\begin{align}
    \left[\ell^\circ_n(t,v),\; u^\circ_n(t,v)\right] = \left[g^{-1}\left(g\left(\theta_{n,l}(t,v)\right) - n^{-1/2}\tilde c_{n,t, v,\alpha}\right), \; g^{-1}\left(g\left(\theta_{n,u}(t,v)\right) + n^{-1/2}\tilde c_{n,t, v,\alpha}\right)\right],
\end{align}
where $\tilde c_{n, t, v, \alpha}$ is such that $P(\tilde Z_1 \leq \tilde c_{n, t, v, \alpha}, \tilde Z_2 \geq -\tilde c_{n, t, v, \alpha}) =(1-\alpha)$, where $(\tilde Z_1, \tilde Z_2)$ follow a mean-zero bivariate normal distribution with cross-fitted estimated covariance as above.

\end{proof}

\begin{proof}[\bfseries{Transformed uniform confidence bands}] 
In practice, there is less variability in survival estimation during the early follow-up period because there are fewer censored units and more event time observations. We thereby construct transformed confidence bands \citep{westling2023inference} with variable width that depends on $t$ to account for the variability in the uncertainty of the estimators over time. Specifically, the transformed $(1-\alpha)$-level confidence bands over $[t_0, t_1]$ are given by 
\begin{align}
    \left[\tilde{\ell}^\circ_n(t,v),\; \tilde{u}^\circ_n(t,v)\right] = \left[g^{-1}\left(g\left(\theta_{n,l}(t,v)\right) - n^{-1/2}\tilde q_{n, v,\alpha}\tilde{\sigma}_{n,l,t,v}\right), \; g^{-1}\left(g\left(\theta_{n,u}(t,v)\right) + n^{-1/2}\tilde q_{n,v,\alpha}\tilde{\sigma}_{n,u,t,v}\right)\right],
\end{align}
where $\tilde{q}_{n,v,\alpha}$ is such that $P(\sup_{t \in [t_0, t_1]} \tilde \xi_{l,t} \leq \tilde q_{n,v,\alpha}, \, \sup_{t \in [t_0, t_1]} \tilde \xi_{u,t} \geq -\tilde q_{n,v,\alpha}) =(1-\alpha)$ where $(\tilde \xi_{l,t}, \tilde \xi_{u,t})$ are simulated paths over $[t_0, t_1]$ from the mean-zero Gaussian process with the covariance estimators \[\tilde \Sigma_{n,v} := \left[\begin{matrix} \tilde \Sigma_{n,l l,v} & \tilde \Sigma_{n,l u,v} \\ \tilde \Sigma_{n,u l,v} & \tilde \Sigma_{n,u u,v} \end{matrix} \right]\] where $\tilde \Sigma_{n,l l,v}:=\Sigma_{n,ll,v}/\{\tilde{\sigma}_{n,l,r,v}\tilde{\sigma}_{n,l,s,v}\}$ and $\tilde \Sigma_{n,u u,v}:=\Sigma_{n,uu,v}/\{\tilde{\sigma}_{n,u,r,v}\tilde{\sigma}_{n,u,s,v}\}$ and $\tilde \Sigma_{n,l u,v}:=\Sigma_{n,l u,v}/\{\tilde{\sigma}_{n,l,r,v}\tilde{\sigma}_{n,u,s,v}\}$. $\Sigma_{n,ll,v}$, $\Sigma_{n,uu,v}$ and $\Sigma_{n,ul,v}$ are the estimated components of the covariance matrix for the correlated Gaussian process $(\d{G}_{n,l,v}$, \,$\d{G}_{n,u,v})$ defined in Section~\ref{sec:inference}. 

\end{proof}

\begin{proof}[\bfseries{Uniform test}]
Denote by $\d{G}_{P,l,v}$ and $\d{G}_{P,u,v}$ the limiting Gaussian processes of $\d{G}_{n,l,v}$ and $\d{G}_{n,u,v}$ respectively. By Theorem~\ref{thm:asy_linear}, under the null hypothesis in Equation~(\ref{eq:uniform_testing}), we have
$$n^{1/2} \left( \sup_{t \in [t_0,t_1]}\left(\theta_{n,l}(t,v)-\theta_0\right), \; \sup_{t \in [t_0,t_1]}-\left(\theta_{n,u}(t,v)-\theta_0\right) \right)$$
converges jointly in distribution to
$$n^{1/2} \left( \sup_{t \in [t_0,t_1]} \d{G}_{P,l,v}(t), \sup_{t \in [t_0,t_1]} -\d{G}_{P,u,v}(t) \right)$$
by the continuous mapping theorem. We define the joint Gaussian process $\d{H}_{P,v}$ to be $\left( \d{G}_{P,l,v}(t),\, \d{G}_{P,u,v}(t)\right)$. The joint Gaussian process $\d{H}_{P,v}$ can be approximated by a correlated Gaussian process with mean zero and estimated covariance matrix $\Sigma_{n,v}$. We outline the testing procedure as follows. 
\begin{enumerate}
    \item Define the test statistic \[T_{n,v}:= \text{max} \left\{ n^{1/2}\sup_{t \in [t_0,t_1]}\left(\theta_{n,l}(t,v)-\theta_0\right), \;  n^{1/2} \sup_{t \in [t_0,t_1]}-\left(\theta_{n,u}(t,v)-\theta_0\right)\right\}.\]
    \item Construct the covariance matrix \[\Sigma_{n,v} := \left[\begin{matrix} \Sigma_{n,l l,v} & \Sigma_{n,l u,v} \\ \Sigma_{n,u l,v} & \Sigma_{n,u u,v} \end{matrix} \right]\] where
        \begin{itemize}
        \item $\Sigma_{n,l l,v}:= \frac{1}{n} \sum_{k=1}^K \sum_{i \in \mathcal{V}_{n, k}}[D^\ast_{n,l, r,v}\left(O_i\right)D^\ast_{n,l, s,v}\left(O_i\right)]$,
        \item $\Sigma_{n,u u,v}:=\frac{1}{n} \sum_{k=1}^K \sum_{i \in \mathcal{V}_{n, k}}[D^\ast_{n,u, r,v}\left(O_i\right)D^\ast_{n,u, s,v}\left(O_i\right)]$ and
        \item $\Sigma_{n,l u,v} = \Sigma_{n,u l} :=\frac{1}{n} \sum_{k=1}^K \sum_{i \in \mathcal{V}_{n, k}}[D^\ast_{n,l, r,v}\left(O_i\right)D^\ast_{n,u, s,v}\left(O_i\right)]$.
        \end{itemize}
    \item Let $\left(  \xi_{n,l,v}, \xi_{n,u,v} \right)$ be the sample paths simulated from the estimated joint Gaussian process with mean zero and estimated covariance matrix $\Sigma_{n,v}$. Define $q_{n,v,\alpha}$ as the $1-\alpha$ quantile of \[\text{max}\left \{\sup_{t \in [t_0,t_1]} \xi_{n,l,v}(t), \; \sup_{t \in [t_0,t_1]} -\xi_{n,u,v}(t)\right\}.\]
    \item Reject the null hypothesis at level $\alpha$ if $T_{n,v} > q_{n,v,\alpha}$.
\end{enumerate}
\end{proof}

\section{Additional details regarding sensitivity analysis for the difference in restricted mean survival time}

We recall that the difference in counterfactual RMST is defined as 
\[ \phi_{c}(t)  := \int_0^t \left[ P_{c}(T(1)>v) - P_{c}(T(0)>v) \right] \, dv = E_c \left[ \min\{ T(1), t \} - \min\{T(0), t\} \right].\]
Analogous to the causal and observed effects in the main paper, under certain conditions, the causal parameter $\phi_{c}(t)$ can be identified through 
\begin{align*}
    \phi_{c}(t) = \phi_{P}(t) & = \int_0^t E [S_P(v \mid 1, X)] \, dv-\int_0^t E [S_P(v \mid 0, X)] \, d v \\
    & = E\left\{E[\min (T, t) \mid A=1, W, U]-E[\min (T, t) \mid A=0, W, U]\right\}
\end{align*}
and the contrast of the observed restricted mean survival times can be written as
\begin{align*}
    \phi_{P}(t) & = \int_0^t E [S_P(v \mid 1, W)]\,  d v-\int_0^t E [S_P(v \mid 0, W)] \, d v \\
    & = E\left\{E[\min (T, t) \mid A=1, W]-E[\min (T, t) \mid A=0, W]\right\}.
\end{align*}

Similar to Proposition~\ref{prop:par_id}, by replacing $I(T>t)$ with $\min(T,t)$, we have the decomposition of $[\phi_P(t) - \phi_c(t)]^2$. We define $h_{c,t}:=E[\min (T, t) \mid A,W,U]=\int_0^t S_c(v \mid A, W, U) dv$ and  $h_{P,t}:=E[\min (T, t) \mid A,W]=\int_0^t S_P(v \mid A, W) \,dv$. 
\begin{prop}\label{prop:rmst}
    We have
\begin{align*}
\left[\phi_P(t) - \phi_c(t)\right]^2 =  \gamma_P(t)\tau_P \left[ \rho_{c, \phi}(t)\right]^2 s_{c, \phi, T}(t) \frac{s_{c,A}}{1-s_{c,A}},
\end{align*}
where 
\begin{align*}
    \rho_{c, \phi}(t) &:= \n{Cor}\left(h_{c,t}-h_{P,t}, \: \alpha_c-\alpha_P\right),\\
    \gamma_P(t) &:= E\left[ \left\{\min (T, t)- h_{P,t}\right\}^2\right],\\
    \tau_P &:= E \left[\left\{\alpha_P\right\}^2\right],\\ 
    s_{c,\phi,T}(t) &:= E\left[\left\{h_{c,t}-h_{P,t}\right\}^2\right] / \gamma_P(t),\\
    s_{c,A} &:= 1 - \frac{E \left[\alpha_P^2\right]}{E \left[\alpha_c^2\right]},
\end{align*}
and where we interpret $0/0$ as 0 if $\gamma_P(t) = 0$ or $\tau_P = 0$. 
Therefore,
\begin{align*}
\phi_{P,l}\left(t, s_{c, \phi, T}(t)s_{c,A}/[1-s_{c,A}]\right) &\leq \phi_c(t) \leq \phi_{P,u}\left(t, s_{c, \phi, T}(t)s_{c,A}/[1-s_{c,A}]\right) \quad \text{ for} \\
    \phi_{P,l}(t, v) &:= \phi_P(t) - \sqrt{|v|\gamma_P(t)\tau_P}, \quad \text{and} \\
    \phi_{P,u}(t,v) &:= \phi_P(t) +  \sqrt{|v|\gamma_P(t)\tau_P}.
\end{align*}
\end{prop}

Proposition~\ref{prop:rmst} yields upper and lower bounds for $\phi_{c}(t)$. Both $\phi_{p}(t)$ and $\gamma_P(t)$ can be written as the integration of the expectation of the conditional survival functions and thus can be estimated by taking integration of the corresponding one-step estimator. Specifically, a natural estimator of $\phi_{p}(t)$ is given by $\int_0^t \theta_{n}(v)\, d v$. As for $\gamma_P(t)$, we can write
\begin{align*} 
\gamma_P(t)  
&:=  E\left\{\min (T, t)-\int_0^t S_P(v \mid A, W) \, d v\right\}^2 \\  
&= E\left\{[\min (T, t)]^2-2 \min (T, t) \int_0^t S_P(v \mid A, W) \, d v + \left[\int_0^t S_P(v \mid A, W)  \, d v\right]^2\right\} \\ 
&= E[\min (T, t)^2]-E\left\{\left[\int_0^t S_P(v \mid A, W) d v\right]^2\right\}\\
&= \int_0^t E[S_P(u \mid A,W)]  2 u \, d u - \int_0^t \int_0^t E[S_P(v \mid A, W) S_P(u \mid A, W)] \, d v \, d u
\end{align*}
where the last step follows by
\begin{align*} 
E\left[\min (T, u)^2\right] &=\int_0^{\infty} \min (t, u)^2 \, d F(t) \\ 
&=\int_0^u t^2 \, d F(t)+\int_u^{\infty} u^2 \, d F(t) \\ 
&=u^2 F(u) -\int_0^u F(t)  \, d (t^2) + u^2 S(u)\\ 
&=u^2+\int_0^u[S(t)-1] 2 t  \, d t \\ 
&=\int_0^u S(t) 2 t  \, d t.
\end{align*}

\begin{prop}\label{prop:eif3}
If there exists $\kappa > 0$ such that $G_P(t \mid a, w) \geq \kappa$ for each $a \in \{0,1\}$ and $P$-almost every $w$ such that $S_P(t \mid a, w) > 0$, then the nonparametric efficient influence functions of $E[S_P(u \mid A,W)]$ and $E[S_P(v \mid A, W) S_P(u \mid A, W)]$ at $P$ are $D_{P,u}^\ast(o) = D_{P,u}(o) - E[S_P(u \mid A,W)]$ and $D_{P,u,v}^\ast(o) = D_{P,u,v}(o) - E[S_P(v \mid A, W) S_P(u \mid A, W)]$ where
\begin{align*}
     & D_{P,u}(o) 
     =  S_P\left(u \mid a, w\right) \left\{ 1 -\frac{I(y \leq u, \delta=1)}{S_P\left(y \mid a, w\right) G_P\left(y \mid a, w\right)} \right. 
       + \left. \int_0^{u \wedge y} \frac{\Lambda_P\left(d u \mid a, w\right)}{S_P\left(u \mid a, w\right) G_P\left(u \mid a, w\right)} \right\} \quad \text{and}\\
     & D_{P,u,v}(o) =  S_P\left(v \mid a, w\right)S_P\left(u \mid a, w\right) \left\{ -\frac{I(y \leq u, \delta=1)}{S_P\left(y \mid a, w\right) G_P\left(y \mid a, w\right)} \right. 
       + \left. \int_0^{u \wedge y} \frac{\Lambda_P\left(d u \mid a, w\right)}{S_P\left(u \mid a, w\right) G_P\left(u \mid a, w\right)} \right\} \\
       &\quad \quad + S_P\left(u \mid a, w\right)S_P\left(v \mid a, w\right) \left\{ -\frac{I(y \leq v, \delta=1)}{S_P\left(y \mid a, w\right) G_P\left(y \mid a, w\right)} \right. 
       + \left. \int_0^{v \wedge y} \frac{\Lambda_P\left(d u \mid a, w\right)}{S_P\left(u \mid a, w\right) G_P\left(u \mid a, w\right)} \right\} \\
     & \quad \quad + S_P\left(u \mid a, w\right)S_P\left(v \mid a, w\right).
\end{align*}
\end{prop}
We can then define a cross-fitted one-step estimator of $\gamma_P(t)$ as
\begin{align*}
\gamma_n(t) =  \int_0^t \frac{1}{n} \sum_{k=1}^K \sum_{i \in \mathcal{V}_{n, k}} D_{n,k,u}\left(O_i\right) \cdot 2 u  \, d u - \int_0^t \int_0^t \frac{1}{n} \sum_{k=1}^K \sum_{i \in \mathcal{V}_{n, k}} D_{n,k,u,v}\left(O_i\right) \, d v  \, d u.
\end{align*}
where $D_{n,k,u}(o)$ and $D_{n,k,u,v}(o)$ are the estimated counterparts of $D_{P,k,u}(o)$ and $D_{P,k,u,v}(o)$. By the delta method, $\phi_{n}(t)$ and $\gamma_n(t)$ are asymptotically linear under similar conditions to those of Theorem~\ref{thm:asy_linear} with influence functions
\begin{align*}
    & D_{P,\phi,t}^\ast(o)= \int_0^t D_{P,\theta,u}^\ast\left(O_i\right)  \, d u \quad \text{and}\\
    & D_{P,\gamma,t}^\ast(o) = \int_0^t D_{P,u}^\ast\left(O_i\right) \cdot 2 u \, d u - \int_0^t \int_0^t  D_{P,u,v}^\ast\left(O_i\right) \, d v\, d u.
\end{align*}
Thus, inference for the bounds $[\phi_{P,l}(t,v), \phi_{P,u}(t,v)]$ can be obtained analogously using the methods described in Section \ref{sec:inference}. Benchmarking of the sensitivity parameter $s_{c, \phi, T}(t)$ using the observed data can also be conducted as described in Section \ref{sec:senspar2}.

\end{appendix}

\end{adjustwidth}

\end{document}